\documentclass{fundam}

\publyear{22}
\papernumber{2162}
\volume{189}
\issue{3-4}

\finalVersionForARXIV

\usepackage{graphicx}
\usepackage{amsmath}
\usepackage{amssymb}
\usepackage{amsfonts}
\usepackage{tikz}
\usepackage{stmaryrd}
\usepackage{mathtools}
\usepackage{multirow}
\usepackage{placeins}
\usepackage{subcaption}
\usepackage{centernot}
\usepackage{verbatim}
\usepackage{ mathrsfs }
\usepackage{longtable}
\usepackage{siunitx}
\usepackage[T1]{fontenc}
\usepackage[ruled,vlined,linesnumbered]{algorithm2e}
\usepackage{ upgreek }
\usepackage[inline]{enumitem}
\usepackage{tabularx}
\usepackage[T1]{fontenc}
\usepackage{longtable}
\usepackage{capt-of}
\usepackage{adjustbox}
\usepackage{bm}
\usepackage[colorinlistoftodos,prependcaption]{todonotes}

\newcommand{\llrr}[1]{\llbracket #1 \rrbracket}

\newcommand{\Ngtzero}{\mathbb{N}^{> 0}}

\newcommand\defeq{\mathrel{\stackrel{\makebox[0pt]{\mbox{\normalfont\tiny def}}}{=}}}

\newcommand{\Bcup}[1]{\textstyle \bigcup_{#1}}

\newcommand{\Bsum}[1]{\textstyle \sum_{#1}}

\newcommand{\lanrang}[1]{\langle #1 \rangle}

\newcommand{\preset}[1]{{}^{\bullet}#1}
\newcommand{\presetstar}[1]{{}^{\circ}#1}

\newcommand{\postset}[1]{#1^{\bullet}}

\newcommand{\dequiv}{\stackrel{\delta}{\equiv}}
\newcommand{\dequivparam}[1]{\stackrel{\delta_{#1}}{\equiv}}

\newcommand{\dequivprime}{\stackrel{\delta'}{\equiv}}

\newcommand{\stabilize}{\textit{Stabilize}}
\newcommand{\dtequiv}{\stackrel{\delta,t}{\equiv}}

\newcommand{\Mdelta}{M^\delta}
\newcommand{\Mdeltaprime}{M^{\delta\prime}}

\newcommand{\CPN}[1]{(P^{#1}, T^{#1},\mathbb{C}^{#1},\mathbb{B}^{#1}, C^{#1}, G^{#1}, W^{#1}, W_I^{#1}, M_0^{#1})}

\newcommand{\deltasCPN}{(P, T,\mathbb{C}^\delta,\mathbb{B}^\delta, C^\delta, G^\delta, W^\delta, W_I^\delta, \Mdelta_0)}

\newcommand{\bempty}{b_\upepsilon}
\newcommand{\specialcell}[2][c]{%
  \begin{tabular}[#1]{@{}r@{}}#2\end{tabular}}

  \newcommand{\specialcelltwo}[2][c]{%
  \begin{tabular}[#1]{@{}l@{}}#2\end{tabular}}

\makeatletter
\def\moverlay{\mathpalette\mov@rlay}
\def\mov@rlay#1#2{\leavevmode\vtop{%
   \baselineskip\z@skip \lineskiplimit-\maxdimen
   \ialign{\hfil$\m@th#1##$\hfil\cr#2\crcr}}}
\newcommand{\charfusion}[3][\mathord]{
    #1{\ifx#1\mathop\vphantom{#2}\fi
        \mathpalette\mov@rlay{#2\cr#3}
      }
    \ifx#1\mathop\expandafter\displaylimits\fi}
\makeatother

\newcommand{\colorsimp}{\textit{set}}

\usepackage[hyphens]{url}

\title{Methods for Efficient Unfolding of Colored Petri Nets}

\author{Alexander Bilgram,\ Peter G. Jensen,\ Thomas Pedersen, \\
            \textrm{\normalsize \bf Ji\v{r}\'{\i} Srba\thanks{Address for correspondence: Aalborg University, Department of Computer Science,
                              Aalborg, Denmark. \newline \newline
                    \vspace*{-6mm}{\scriptsize{Received April 2022; \ accepted January 2023.}}}, \  Peter H. Taankvist}
 \\
Aalborg University\\
Department of Computer Science,  Aalborg, Denmark\\
pgj@cs.aau.dk, srba@cs.aau.dk 
}

\runninghead{A. Bilgram et. al.}{Methods for Efficient Unfolding of Colored Petri Nets}

\begin{document}

\maketitle

\vspace*{-6mm}
\begin{abstract}
Colored Petri nets offer a compact and user friendly representation
of the traditional Place/Transition (P/T) nets and colored nets with finite color ranges
can be unfolded into the underlying P/T nets, however, at the expense
of an exponential explosion in size.
We present two novel techniques based on static analysis
in order to reduce the size of unfolded colored nets. The first method
identifies colors that behave equivalently and groups them into equivalence
classes, potentially reducing the number of used colors. The second method
overapproximates the sets of colors that can appear in places and excludes
colors that can never be present in a given place.
Both methods are complementary and the combined approach allows us to
significantly reduce the size of multiple colored Petri nets from the
Model Checking Contest benchmark. We compare the performance of our unfolder
with state-of-the-art techniques implemented in the tools
MCC, Spike and ITS-Tools, and while our approach is
competitive w.r.t.
unfolding time, it also outperforms the existing approaches both
in the size of unfolded nets as well as in the number of answered model
checking queries from the 2021 Model Checking Contest.
\end{abstract}

\section{Introduction} \label{sec:Introduction}

Petri nets~\cite{CarlPetriNets}, also known as P/T nets,
are a powerful modelling formalism supported by a rich family of
verification techniques~\cite{MurataPaper}.
However, P/T nets often become too large and
incomprehensible for humans to read. Therefore, colored Petri nets
(CPN)~\cite{CPNJensen} were introduced to allow for high level
modelling of distributed systems.
In CPNs, each place is assigned a color domain and each token in that place
has a color from its domain. Arcs have expressions that define
what colored tokens to consume or produce, and transitions have
guard expressions that restrict transition enabledness.

A CPN can be translated into an equivalent P/T net, provided
that every color domain is finite, through a process
called \textit{unfolding}.  This allows us to use efficient verification
tools already developed for P/T nets. When unfolding a CPN, each place
is unfolded into a new place for each color that a token can take
in that place; a naive approach is to create a new place for each color
in the color domain of the place. Transitions are unfolded
such that each binding of variables to colors, satisfying the guard,
is unfolded into a new transition copy in the unfolded net.
The size of an unfolded net can be exponentially larger than the
colored net and the unfolding process
therefore requires optimizations
in order to finish in realistic time and memory. Several types of improvements
were proposed that analyse transition guards and
arc expressions~\cite{MCCUnfolder, PatternPaper, IDDUnfolder}.
However, even with these optimizations, there still exist CPNs that
cannot be unfolded using the existing tools. As an example, the largest instances
of the nets \textit{FamilyReunion}~\cite{FamilyReunionNet, FamilyReunionPaper}
and \textit{DrinkVendingMachine}~\cite{DrinkVendingMachineNet,
DrinkVendingMachinePaper} from the Model Checking Contest~\cite{mcc:2021}
have not yet been unfolded in the competition setup.

\medskip
We propose two novel methods for statically analysing a CPN
to reduce the size of the unfolded P/T net. The first method called
\textit{color quotienting} uses the fact that sometimes tokens with different colors
can be indistinguishable in the sense that they generate bisimilar behaviour. If such colors exist in the
net, we can create equivalence classes that represent the colors with
similar behaviour. As such, we can reduce the amount of colors
that we need to consider when unfolding.
The second method called \textit{color approximation} overapproximates
which colors can possibly be present in any given place such that
we only unfold places for the colors that can exist. This method also
allows for invalidating bindings that are dependent on unreachable
colors, thus reducing the amount of transitions that are unfolded.

Our two methods are implemented in the model checker TAPAAL~\cite{TAPAALTool, verifypnPaper} and an extensive experimental evaluation
shows convincing performance compared to the state-of-the-art tools for CPN unfolding.

\paragraph{Related work.}
Heiner et al.~\cite{PatternPaper}  analyse the arc and guard expressions
to reduce the amount of bindings
by collecting \textit{patterns}.
The pattern analysis is implemented in the tool
Snoopy~\cite{SnoopyTool} and our color approximation method
further extends this method.
In~\cite{IDDUnfolder} the same authors present a technique
for representing the patterns as Interval Decision Diagrams. This
technique is used in the tools Snoopy~\cite{SnoopyTool},
MARCIE~\cite{MARCIETool} and Spike~\cite{SpikeTool}
and performs better compared to~\cite{PatternPaper}; it also allows to unfold a superset
of colored nets compared to the format adopted by the Model Checking Contest
benchmark~\cite{mcc:2021}.

In~\cite{MCCUnfolder} Dal-Zilio describes a method (part of the unfolder MCC) called
\textit{stable places}. A stable place is a place that never changes
from the initial marking, i.e. every time a token is consumed
from this place an equivalent token of the same color is added to the place.
This method is especially efficient on the net BART from the
Model Checking Contest~\cite{mcc:2021}, however,
it does not detect places that deviate even a little from the initial marking.
Our color approximation method includes a more general form of the stable places.
In the unfolder MCC~\cite{MCCUnfolder}, a \textit{component analysis}
is introduced and it detects if a net consists of a number of
copies of the same component. MCC is used in the TINA
toolchain~\cite{TINAtool} and to our knowledge in the latest release of the
LoLA tool~\cite{LoLA2}.
GreatSPN~\cite{GSPNPaper} is another tool for unfolding CPNs,
however, in~\cite{MCCUnfolder} it is demonstrated that MCC is
able to greatly outperform GreatSPN and as such we omit GreatSPN from later experiments.

\eject
ITS-Tools~\cite{ITSToolPaper} has an integrated unfolding engine.
The tool uses a technique of \textit{variable symmetry identification},
in which it is analyzed whether variables $x$ and $y$ are permutable
in a binding. 
Furthermore, they use stable places during the binding and they apply
analysis to choose the binding order of parameters to simplify
false guards as soon as possible.
After unfolding, ITS-Tools applies further post-unfolding reductions
that remove orphan places/transitions and behaviourally equivalent
transitions~\cite{YannCorrespondence}. Our implementation includes a variant
of the symmetric variables reduction as well.
In \cite{YannSymmetryDetection} Thierry-Mieg et al. present a
technique for automatic detection of symmetries in high level Petri nets
used to construct symbolic reachability graphs in the GreatSPN tool.
This detection of symmetries is reminiscent of the color quotienting
method presented in this paper, although our color quotienting method
is used for unfolding the colored Petri net instead of symbolic model\linebreak checking.

\medskip
In~\cite{KlostergaardThesis} Klostergaard presents a simple unfolding method
implemented in TAPAAL~\cite{TAPAALTool, verifypnPaper},
which is the base of our implementation. The implementation is efficient
but there are several nets which it cannot unfold.
Both unfolding methods 
introduced in this paper
are advanced static analyses techniques and we observe that the above
mentioned techniques, except symmetric variables and component analysis,
are captured by color approximation and/or color quotienting.

\medskip
This paper is an extended version of the conference paper~\cite{BJPST:RP:2021} with full proofs,
complete definitions and additional examples and last but not least a substantial reimplementation
of the methods in the tool TAPAAL with improved experimential results compared to~\cite{BJPST:RP:2021}.

\section{Preliminaries} \label{sec:Preliminaries}
Let $\mathbb{N}^{> 0}$ be the set of positive integers
and $\mathbb{N}^{0}$ the set of nonnegative integers.
A Labeled Transition System (LTS) is a triple $(Q,Act, \xrightarrow{})$
where $Q$ is a set of states, $Act$ is a finite, nonempty set of actions,
and  $\xrightarrow{} \subseteq Q \times Act \times Q$
is the transition relation.
We write $s \xrightarrow{a}$ if there is $s' \in Q$ such that $s \xrightarrow{a} s'$
and $s \not\xrightarrow{a}$ if there is no such $s'$.
A binary relation $R$ over the set of states of an LTS is a \emph{bisimulation}
iff for every $(s_1,s_2) \in R$ and $a \in Act$ it holds that
if $s_1 \xrightarrow{a} s_1'$ then there is a
transition $s_2 \xrightarrow{a} s_2'$ such that $(s_1',s_2') \in R$, and
if $s_2 \xrightarrow{a} s_2'$ then there is a transition $s_1 \xrightarrow{a} s_1'$ such that $(s_1',s_2') \in R$.
Two states $s$ and $s'$ are \emph{bisimilar}, written $s \sim s'$,
iff there is a bisimulation $R$ such that $(s,s') \in R$.

A finite \emph{multiset} over some nonempty set $A$
is a collection of elements from $A$ where each element occurs
in the multiset a finite amount of times;
a multiset $S$ over a set $A$ can be identified with
a function $S: A \xrightarrow{} \mathbb{N}^{0}$
where $S(a)$ is the number of occurrences of element $a \in A$
in the multiset $S$. We shall represent multisets by a formal sum
$\sum_{a \in A} S(a)'(a)$
such that e.g. $1'(x) + 2'(y)$ stands for a multiset containing
one element $x$ and two elements $y$.
We assume the standard multiset operations of membership ($\in$),
inclusion ($\subseteq$), equality ($=$), union ($\uplus$),
subtraction ($\setminus$) and by $|S|$ we denote the cardinality of $S$ (including the repetition of elements).
By $\mathcal{S}(A)$ we denote the set of all multisets over
the set $A$.

\medskip
Finally, we also define the function $\colorsimp$
as a way of reducing multisets of colors to sets of colors given by
   $ \colorsimp(S) \defeq \{a \ | \ a \in S\}$
where $\colorsimp (S)$ is the set of all colors with at least one occurrence\linebreak  in $S$.

\subsection{Colored Petri nets} \label{sec:CPN}
Colored Petri nets (CPN) are an extension of traditional P/T nets
introduced by Kurt Jensen~\cite{CPNJensen} in 1981. In CPNs, places are associated with
color domains where colors represent the values of tokens.
Arc expressions describe what colors to consume and add to places
depending on a given binding (assignment of variables to colors).
Transitions may contain guards restricting which bindings are valid.
There exist several different definitions of CPNs from the
powerful version defined in~\cite{JensenCPNBook2} that includes the
ML language for describing arcs expressions and guards to less powerful ones
such as the one used in the Model Checking Contest~\cite{mcc:2021}.
We shall first give an abstract definition of a CPN.

\begin{definition}
A colored Petri net is a tuple $\mathcal{N} = \CPN{}$
where
\begin{enumerate}
    \item $P$ is a finite set of places,
    \item $T$ is a finite set of transitions such that $P \cap T = \emptyset$,
    \item $\mathbb{C}$ is a nonempty set of colors,
    \item $\mathbb{B}$ is a nonempty set of bindings,
    \item $C: P \xrightarrow{} 2^\mathbb{C} \setminus \emptyset$ is a place color type function,
    \item $G: T \times \mathbb{B} \xrightarrow{} \{\mathit{true}, \mathit{false}\}$ is a guard evaluation function,
    \item $W\!: ((P \times T) \cup (T \times P)) \times \mathbb{B} \xrightarrow{} \mathcal{S}(\mathbb{C})$ is an arc evaluation function such that $\colorsimp(W((p,t),b))\!\subseteq C(p)$ and $\colorsimp(W((t,p),b)) \subseteq C(p)$ for all $p \in P$, $t \in T$ and $b \in \mathbb{B}$,
    \item $W_I: P \times T \xrightarrow{} \Ngtzero \cup \{\infty\}$ is an inhibitor arc weight function, and
    \item $M_0$ is the initial marking where a marking $M$ is a function $M\!: \!P \xrightarrow{} \mathcal{S}(\mathbb{C})$
                   such that $\colorsimp (M(p))\! \subseteq C(p)$ for all $p \in P$.
\end{enumerate}
\end{definition}

In this definition, we assume an abstract set of bindings $\mathbb{B}$, representing
the concrete configurations a net can be in. For example, colored Petri nets often
use variables on the arcs and in the guards, and these variables can be assigned
concrete colors. Such a variable assignment is then referred to as a binding.
Notice that $G$ and $W$ are semantic functions which are in different
variants of CPN defined by a concrete syntax. These functions take as an argument
a binding and for this binding return either whether the guard is true or false,
or in case of $W$ the function returns the multiset of tokens that should be consumed/produced
when a transition is fired.

The set of all markings on a CPN $\mathcal{N}$ is denoted by
$\mathbb{M}(\mathcal{N})$. In order to avoid the use of partial functions,
we allow $W((p,t),b) = W((t,p),b) = \emptyset$ and $W_I(p,t) = \infty$,
meaning that if the arc evaluation function returns the empty multiset
then the arc has no effect on transition firing and if the inhibitor arc
function returns infinity then it never inhibits the connected transition.

\medskip
Let $\mathcal{N} = \CPN{}$ be a fixed CPN for the rest of this section.
Let $B(t) \defeq \{ b \in \mathbb{B} \ | \ G(t,b) = \textit{true} \}$ be the
set of all bindings that satisfy the guard of transition $t \in T$.
Let $\ell : T \xrightarrow{} Act$ be a transition labeling function.
The semantics of a CPN $\mathcal{N}$ is defined as an LTS
$L(\mathcal{N}) = (\mathbb{M}(\mathcal{N}),Act,\xrightarrow{})$ where
$\mathbb{M}(\mathcal{N})$ is the set of states defined as all markings on $\mathcal{N}$,
$Act$ is the set of actions, and
$M \xrightarrow{a} M'$ iff there exists $t \in T$ where $\ell(t) = a$ and there is
$b \in B(t)$ such that
    \begin{align*}
        &W((p,t),b) \subseteq M(p)  \text{ and }
        W_I(p,t) > |M(p)| \text{ for all } p \in P , \text{ and}\\
        &M'(p) = (M(p) \setminus W((p,t),b)) \uplus W((t,p),b) \text{ for all } p \in P.
    \end{align*}
We denote the firing of a transition $t \in T$ in marking $M$ reaching
$M'$ as $M \xrightarrow{t} M'$.  Let $\xrightarrow{} = \Bcup{t \in T}\xrightarrow{t}$ and let $\xrightarrow{}^*$ be the reflexive and transitive closure of $\xrightarrow{}$.

\begin{remark} \label{remark:Finite}
To reason about model checking of CPNs, we need to have
a finite representation of colored nets that can be passed as an input to an
algorithm. One way to enforce such a representation is to assume that
all color domains are finite and the semantic
functions $C$, $G$, $W$ and $W_I$ are effectively computable.
\end{remark}

Finally, let us define the notion of postset and preset of $p \in P$
as $\postset{p} = \{ t \in T \mid \exists b \in \mathbb{B}.\ W((p,t),b) \neq \emptyset \}$
and $\preset{p} = \{ t \in T \mid \exists b \in \mathbb{B}.\ W((t,p),b) \neq \emptyset \}$.
Similarly, for a transition $t \in T$ we define
$\postset{t} = \{ p \in P \mid \exists b \in \mathbb{B}.\ W((t,p),b) \neq \emptyset \}$
and $\preset{t} = \{ p \in P \mid \exists b \in \mathbb{B}.\ W((p,t),b) \neq \emptyset \}$. We also define the preset of inhibitor arcs as $\presetstar{t} = \{ p \in P \ | \ W_I(p,t) \neq \infty \}$.

\subsection{P/T nets}
A Place/Transition (P/T) net is a CPN $\mathcal{N} = \CPN{}$ with
one color $\mathbb{C} = \{ \bullet \}$ and only one binding
$\mathbb{B} = \{\bempty\}$ such that every guard evaluates to true
i.e. $G(t,\bempty) = \textit{true}$ for all $t\in T$ and every arc evaluates to a multiset over
$\{\bullet\}$ i.e.
$W((p,t),\bempty) \in \mathcal{S}(\{ \bullet \})$
and $W((t,p),\bempty) \in \mathcal{S}(\{ \bullet \})$ for all $p \in P$ and $t \in T$.

\subsection{Integer colored Petri nets} \label{Sec:IntegerCPN}

An integer CPN (as used for example in the Model Checking Contest~\cite{mcc:2021})
is a CPN $$\mathcal{N} = \CPN{}$$
where all colors are integer products i.e. $\mathbb{C} = \Bcup{k \geq 1}(\mathbb{N}^0)^{k}$. We use
interval ranges to describe sets of colors such that
a tuple of ranges $([a_1,b_1],...,[a_k,b_k])$ where $a_i,b_i \in \mathbb{N}^0$
for $i, 1 \leq i \leq k$, describes the set of colors
    $\{(c_1,...,c_k) \ | \ a_i \leq c_i \leq b_i \text{ for all } 1 \leq i \leq k \}$.
If the interval upper-bound is smaller than the lower-bound,
the interval range denotes the empty set and by $[a]$ we denote the singleton interval $[a,a]$.
As an example, consider the place color type of some place $p$ as $C(p) = \llrr{([1,2],[6,7])}$
describing the set of colors $\{(1,6),(1,7),(2,6),(2,7)\}$. For notational convenience,
we sometimes omit the semantic paranthesis and simply write $([1,2],[6,7])$ instead of $\llrr{([1,2],[6,7])}$.

We use the set of variables $\mathcal{V}=\{x_1,...,x_n\}$ to represent colors. Variables can be
present on arcs and in guards. A binding $b: \mathcal{V} \xrightarrow{} \mathbb{C}$
assigns colors to variables.
We write $b \equiv \langle x_1=c_1,...,x_n=c_n \rangle$
for a binding where $b(x_i)=c_i$ for all $i$, $1 \leq i \leq n$.
We now introduce the syntax of arc/guard expressions and its intuitive semantics
by an example. 

Figure~\ref{fig:CPNexample} shows an integer CPN where places (circles)
are associated with ranges. 
The initial marking contains
five tokens (two of color $0$ and three of color $2$) in $p_1$ and
two tokens of color $5$ in place $p_2$.
There is a guard on transition $t$ (rectangle)
that compares $x$ with the integer $1$ and restricts the valid bindings.
We can see that the arc from $t$ to $p_3$
creates a product of the integers $x$ and $y$, where the value of $x$ is decremented by one.
We assume that all ranges are cyclic, meaning that the predecessor of $0$ in the color set $A$ is $2$.
Figure~\ref{fig:CPNexample} also shows an example of transition firing.
Markings are written as formal sums showing how many tokens of what colors
are in the different places. The transition $t$ can fire only once, as the
inhibitor arc (for unlabelled inhibitor arcs we assume the default weight 1) from
place $p_3$ to transition $t$ inhibits the second transition firing.

\begin{figure}[!t]
\centering
\begin{subfigure}[b]{\textwidth}
\begin{center}
    \begin{adjustbox}{width=.8\linewidth} 

\tikzset{every picture/.style={line width=0.75pt}} 

\begin{tikzpicture}[x=0.75pt,y=0.75pt,yscale=-1,xscale=1]

\draw   (150.5,214.25) .. controls (150.22,204.17) and (158.17,196) .. (168.25,196) .. controls (178.33,196) and (186.72,204.17) .. (187,214.25) .. controls (187.28,224.33) and (179.33,232.5) .. (169.25,232.5) .. controls (159.17,232.5) and (150.78,224.33) .. (150.5,214.25) -- cycle ;
\draw  [color={rgb, 255:red, 0; green, 0; blue, 0 }  ,draw opacity=1 ] (440.5,177.25) .. controls (440.22,167.17) and (448.17,159) .. (458.25,159) .. controls (468.33,159) and (476.72,167.17) .. (477,177.25) .. controls (477.28,187.33) and (469.33,195.5) .. (459.25,195.5) .. controls (449.17,195.5) and (440.78,187.33) .. (440.5,177.25) -- cycle ;
\draw   (150.5,135.25) .. controls (150.22,125.17) and (158.17,117) .. (168.25,117) .. controls (178.33,117) and (186.72,125.17) .. (187,135.25) .. controls (187.28,145.33) and (179.33,153.5) .. (169.25,153.5) .. controls (159.17,153.5) and (150.78,145.33) .. (150.5,135.25) -- cycle ;
\draw  [color={rgb, 255:red, 0; green, 0; blue, 0 }  ,draw opacity=1 ][fill={rgb, 255:red, 0; green, 0; blue, 0 }  ,fill opacity=1 ] (313.99,160.59) -- (313.69,194.18) -- (303.7,194.09) -- (304,160.5) -- cycle ;
\draw    (187,135.25) -- (300.13,169.63) ;
\draw [shift={(303,170.5)}, rotate = 196.9] [fill={rgb, 255:red, 0; green, 0; blue, 0 }  ][line width=0.08]  [draw opacity=0] (8.93,-4.29) -- (0,0) -- (8.93,4.29) -- cycle    ;
\draw    (188,214.25) -- (300.1,184.27) ;
\draw [shift={(303,183.5)}, rotate = 525.03] [fill={rgb, 255:red, 0; green, 0; blue, 0 }  ][line width=0.08]  [draw opacity=0] (8.93,-4.29) -- (0,0) -- (8.93,4.29) -- cycle    ;
\draw    (308.84,177.34) .. controls (346.24,148.1) and (404.61,148.47) .. (439.87,170.15) ;
\draw [shift={(442,171.5)}, rotate = 213.31] [fill={rgb, 255:red, 0; green, 0; blue, 0 }  ][line width=0.08]  [draw opacity=0] (8.93,-4.29) -- (0,0) -- (8.93,4.29) -- cycle    ;
\draw    (320.23,185.66) .. controls (361.76,209.18) and (404.78,201.16) .. (443,184.5) ;
\draw [shift={(317.69,184.18)}, rotate = 30.9] [color={rgb, 255:red, 0; green, 0; blue, 0 }  ][line width=0.75]      (0, 0) circle [x radius= 3.35, y radius= 3.35]   ;

\draw (346,135.4) node [anchor=north west][inner sep=0.75pt]    {$1'( x-1,y)$};
\draw (224,127.4) node [anchor=north west][inner sep=0.75pt]    {$1'( x)$};
\draw (224,203.4) node [anchor=north west][inner sep=0.75pt]    {$1'( y)$};
\draw (304,201.4) node [anchor=north west][inner sep=0.75pt]    {$t$};
\draw (290,141.4) node [anchor=north west][inner sep=0.75pt]    {$x< 1$};
\draw (111,114) node [anchor=north west][inner sep=0.75pt]   [align=left] {[$A$]};
\draw (113,194) node [anchor=north west][inner sep=0.75pt]   [align=left] {[$B$]};
\draw (444,134) node [anchor=north west][inner sep=0.75pt]   [align=left] {[$AB$]};
\draw (112,133.4) node [anchor=north west][inner sep=0.75pt]   {$ \begin{array}{l}
2'( 0)\\
3'( 2)
\end{array}$};
\draw (112,213.4) node [anchor=north west][inner sep=0.75pt]     {$2'( 5)$};
\draw (497,124) node [anchor=north west][inner sep=0.75pt]   [align=left] {\textbf{Declarations:}\\color set $A=([0,2])$\\color set $B=([4,5])$\\color set $AB=A \times B$ \\variable $x: A$\\variable $y: B$};
\draw (162,161.4) node [anchor=north west][inner sep=0.75pt]    {$p_{1}$};
\draw (162,241.4) node [anchor=north west][inner sep=0.75pt]    {$p_{2}$};
\draw (451,204.4) node [anchor=north west][inner sep=0.75pt]    {$p_{3}$};
\draw (372,205.4) node [anchor=north west][inner sep=0.75pt]    {$$};

\end{tikzpicture}

    \end{adjustbox}

    \begin{adjustbox}{width=0.95\linewidth} 
\tikzset{every picture/.style={line width=0.75pt}} 
\begin{tikzpicture}[x=0.75pt,y=0.75pt,yscale=-1,xscale=1]
\draw (16,158.4) node [anchor=north west][inner sep=0.75pt]  [font=\large]
{$p_{1} :( 2'( 0) +3'( 2)) \ +\ p_{2} :2'( 5)$};
\draw (268,151.4) node [anchor=north west][inner sep=0.75pt]  [font=\Large]  {$\xrightarrow{t}$};
\draw (313,158.4) node [anchor=north west][inner sep=0.75pt]  [font=\large]
{$p_{1} :( 1'( 0) +3'( 2))  \ +\ p_{2} :1'( 5)  \ +\ p_{3} :( 1'(( 2,5)))$};
\end{tikzpicture}

    \end{adjustbox}       %
\end{center}
\caption{Integer CPN and transition firing under the binding $\langle x=0,y=5\rangle $}
\label{fig:CPNexample}
\end{subfigure}

\begin{subfigure}[b]{\textwidth}
    \centering
    \begin{adjustbox}{width=0.5\linewidth} 
    \input{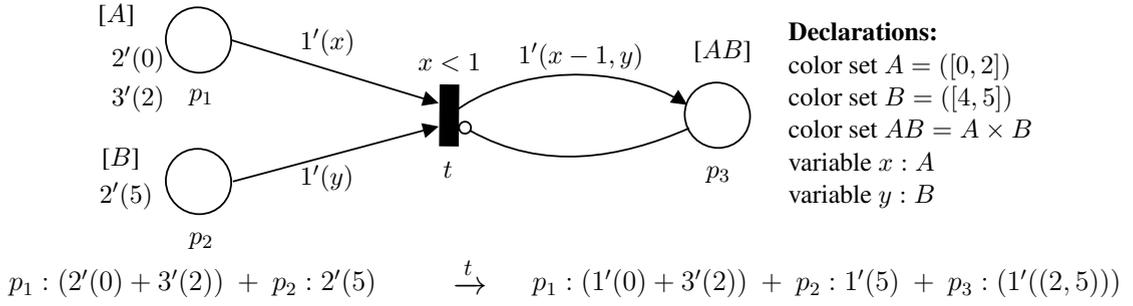}
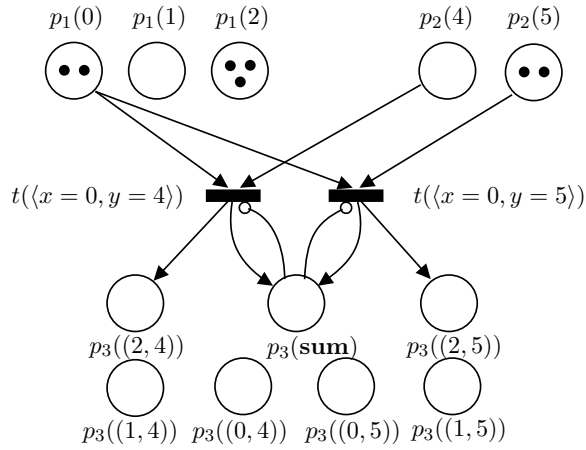
    \end{adjustbox}       %
\caption{Unfolding of CPN from Figure~\ref{fig:CPNexample} }
   \label{fig:naiveUnfolding}
\end{subfigure}
   \caption{A CPN and its Unfolding to a P/T net}
\end{figure}

\medskip
We shall now present the formal syntax of arc expressions and guards.
In integer CPNs,
each arc $(P \times T) \cup (T \times P)$ excluding inhibitor arcs is assigned an arc expression $\alpha$ given by the syntax:
\begin{align*}
    &\alpha ::= n'(\tau_1,...,\tau_k) \ |\ \alpha_1 \pm \alpha_2\ | \ n \cdot \alpha \\
    &\tau ::= c \ | \ x \ | \ x \pm s\
\end{align*}
where $c \in \mathbb{C}$, $x \in \mathcal{V}$, $s \in \Ngtzero$, $n \in \Ngtzero$ and $\pm ::= + \ | \ -$.
The $\alpha$ expressions allow us to make tuples consisting of $\tau$ expressions
that can be combined using multiset union ($+$) and subtraction ($-$), or multiplied by $n$ (creating $n$
copies of the tuple). The expressions of type $\tau$ are called simple arc expressions.
The semantics of arc expressions is straightforward and demonstrated by the following example.
\begin{example}
Let $a$ be an arc annotated by the arc expression $1'(x-1) + 1'(y + 1) + 1'(z)$ and
let $b_1 = \lanrang{x=3,y=3,z=1}$ and $b_2 = \lanrang{x=1,y=2,z=2}$ be bindings with integer range $([1,3])$ over
the variables $x, y$ and $z$. The CPN semantics of the arc $a$ is defined as the multiset
where $W(a,b_1) = 1'(2) + 2'(1)$ since the colors are cyclic in nature such that $3+1 = 1$
and $W(a,b_2) = 2'(3) + 1'(2)$ because $1-1 = 3$.
\end{example}

Guards in integer CPNs are expressed by the following syntax:
\begin{equation*}
    \gamma ::= \mathit{true}\ |\ \mathit{false}\ |\ \neg\gamma\ |\ \gamma_1 \wedge \gamma_2\ |\ \gamma_1 \vee \gamma_2\ |\ \alpha_1 = \alpha_2\ |\ \alpha_1 \neq \alpha_2\ |\ \tau_1 \bowtie \tau_2 \
\end{equation*}
where $\bowtie \ ::= <\ |\ \leq\ |\ >\ |\ \geq\ |\ =\ |\ \neq$
and where $\alpha_1$ and $\alpha_2$ are arc expressions and $\tau_1$ and $\tau_2$ are simple arc expressions.
The semantics of guards (their evaluation to true or false
in a given binding) is
also straightforward and demonstrated by an example.
\begin{example}
Let $g = (x > 2 \wedge y = 2) \vee z + 2 = 3$ be a guard on a transition $t$
and let $b_1 = \lanrang{x=3,y=3,z=1}$ and $b_2 = \lanrang{x=1,y=2,z=2}$ be bindings with range $([1,3])$ over
the variables $x, y$ and $z$. Then $G(t,b_1) = \textit{true}$ and $G(t,b_2) = \textit{false}$.
\end{example}


The Model Checking Contest~\cite{mcc:2021} further includes color types called dots and cyclic enumerations which are excluded from these definitions as these can be trivially translated to tuples of integer ranges. The color type dot, $\{\bullet\}$, is represented
by the color domain $([1])$ and a cyclic enumeration with elements $\{e_1,e_2,...,e_n\}$ is encoded as the integer range  $([1,n])$
corresponding to the indices of the cyclic enumeration. Furthermore, the contest uses the \textit{.all} expression,
which creates one of each color in the color domain. For example
if $A = ([0,2])$ then $A.all = 1'(1) + 1'(2) + 1'(3)$.

All examples in this paper are expressed in integer CPN syntax.

\subsection{Unfolding of CPNs}
CPNs with finite color domains can be \textit{unfolded} into an equivalent P/T net~\cite{JensenCPNBook}.
Each place $p$ is unfolded into $|C(p)|$ places, a transition is made for each legal binding
and we translate the multiset of colors on the arc to a multiset over $\bullet$.
We now provide a formal definition of unfolding in our syntax, following the approach
from~\cite{KlostergaardThesis,LattePaper} that also consider inhibitor arcs.

For each place connected to an inhibitor arc, we create a fresh summation place that contains
the sum of tokens across the rest of the unfolded places. The summation places are created
to ensure that inhibitor arcs work correctly after unfolding.

\begin{definition}[Unfolding]\label{Def:NaiveUnfolding}
Let $\mathcal{N} = \CPN{}$ be a colored Petri net. The unfolded P/T net
$\mathcal{N}^u = \CPN{u}$ is given by
\begin{enumerate}
 \itemsep=0.94pt
    \item $P^u = \{p(c) \ | \ p \in P \wedge c \in C(p)\} \cup \{ p({\textbf{sum}}) \ | \ t \in T, p \in \presetstar{t}\},$
    \item $T^u = \bigcup_{t \in T}\bigcup_{b \in B(t)} t(b),$
    \item $\mathbb{C}^u = \{\bullet\}$,
    \item $\mathbb{B}^u = \{\bempty\}$,
    \item $C^u(p(c)) = \{\bullet\}$ for all $p(c)\in P^u$,
    \item $G^u(t(b),\bempty) = \textit{true}$ for all $t(b) \in T^u$,
    \item $W^u((p(c),t(b)),\bempty) = W((p,t), b)(c)'(\bullet) \text{ and }
    W^u((t(b),p(c)),\bempty) = \\ W((t,p), b)(c)'(\bullet)$ for all \ $p(c) \in P^u$ and \ $t(b) \in T^u$, \textit{and} \\
    $W^u((p(\textbf{sum}),t(b)),\bempty) = |W((p,t), b)|'(\bullet) \text{ and }  W^u((t(b),p(\textbf{sum})),\bempty)
  = \\ |W((t,p), b)|'(\bullet)$ for all $p(\textbf{sum}) \in P^u$ and $t(b) \in T^u$,
    \item $W_I^u(p(\textbf{sum}),t(b)) = W_I(p,t)$ for all $p(\textbf{sum}) \in P^u$ and $t(b) \in T^u$, and
    \item $M_{0}^u(p(c)) = M_0(p)(c)'(\bullet)$ for all $p(c) \in P^u$ \textit{ and }\\
    $M_{0}^u(p(\textbf{sum})) = |M_0(p)|'(\bullet)$ for all $p(\textbf{sum}) \in P^u$
\end{enumerate}
where $p(\textbf{sum})$ denotes the sum of all tokens regardless of the color for the place $p$.
\end{definition}\normalsize

Consider the CPN in Figure~\ref{fig:CPNexample}. The unfolded version of this can be seen
in Figure~\ref{fig:naiveUnfolding}. We notice that each place of the CPN is unfolded to a fresh new place for every color in the color type of the place as well as a \textbf{{sum}} place for $p_3$. Additionally, the transition is unfolded to a new transition for each legal binding.

The theorem showing that the unfolded net is bisimilar to the original CPN
was proved in~\cite{KlostergaardThesis,LattePaper}; we only added a small optimization on
the summation places.

\begin{theorem}[\cite{KlostergaardThesis,LattePaper}] \label{theorem:BisimilarUnfolding}
Given a colored Petri net $\mathcal{N} = \CPN{}$ and the unfolded P/T net $\mathcal{N}^u = \CPN{u}$, it holds that
$M_0 \sim M_0^u$ with the labeling function $\ell(t(b)) = t$ for all $t(b) \in T^u$.
\end{theorem}

\section{Color quotienting} \label{Sec:Partitioning}
Unfolding a CPN without any further analysis will often lead to many unnecessary places and transitions.
We shall now present our first technique that allows to group equivalently behaving colors into
equivalence classes in order to reduce the number of colors and hence also to reduce
the size of the unfolded net.

\medskip
As an example consider the CPN in Figure~\ref{fig:UnstablePartition}, the unfolded version of this net
adds five places for both $p_1$ and $p_2$. However, we see that in $p_1$ all colors greater than or equal
to $3$ behave exactly the same throughout the net and can thus be represented by a single color.
We can thus \textit{quotient} the CPN by \textit{partitioning} the color domain of each place into a
number of \textit{equivalence classes} of colors such that the colors behaving equivalently are represented
by the same equivalence class.
Using this approach, we can construct a bisimilar CPN seen in Figure~\ref{fig:partitionUnfolding}
where the color $3$ now represents all colors greater than or equal to $3$.

\medskip
Such a reduction in the number of colors is possible to include already during the design of a CPN model,
however, the models may look less intuitive for human modeller or the nets can be auto-generated
and hence contain redundant/equivalent colors as observed in the benchmark of CPN models from
the annual Model Checking Contest benchmark~\cite{mcc:2021}.

\begin{figure}[th]
\centering
\begin{subfigure}[b]{\textwidth}
    \centering
    \begin{adjustbox}{width=0.7\linewidth} 
    \tikzset{every picture/.style={line width=0.75pt}} 

\begin{tikzpicture}[x=0.75pt,y=0.75pt,yscale=-1,xscale=1]

\draw   (100,151.75) .. controls (100,141.95) and (108.06,134) .. (118,134) .. controls (127.94,134) and (136,141.95) .. (136,151.75) .. controls (136,161.55) and (127.94,169.5) .. (118,169.5) .. controls (108.06,169.5) and (100,161.55) .. (100,151.75) -- cycle ;
\draw    (136,151.75) -- (215,151.51) ;
\draw [shift={(218,151.5)}, rotate = 539.8299999999999] [fill={rgb, 255:red, 0; green, 0; blue, 0 }  ][line width=0.08]  [draw opacity=0] (8.93,-4.29) -- (0,0) -- (8.93,4.29) -- cycle    ;
\draw  [fill={rgb, 255:red, 0; green, 0; blue, 0 }  ,fill opacity=1 ] (219,135) -- (227,135) -- (227,168.5) -- (219,168.5) -- cycle ;
\draw   (309,150.75) .. controls (309,140.95) and (317.06,133) .. (327,133) .. controls (336.94,133) and (345,140.95) .. (345,150.75) .. controls (345,160.55) and (336.94,168.5) .. (327,168.5) .. controls (317.06,168.5) and (309,160.55) .. (309,150.75) -- cycle ;
\draw    (228,150.75) -- (306,150.75) ;
\draw [shift={(309,150.75)}, rotate = 180] [fill={rgb, 255:red, 0; green, 0; blue, 0 }  ][line width=0.08]  [draw opacity=0] (8.93,-4.29) -- (0,0) -- (8.93,4.29) -- cycle    ;
\draw  [fill={rgb, 255:red, 0; green, 0; blue, 0 }  ,fill opacity=1 ] (427,135) -- (435,135) -- (435,168.5) -- (427,168.5) -- cycle ;
\draw    (346,150.75) -- (424,150.75) ;
\draw [shift={(427,150.75)}, rotate = 180] [fill={rgb, 255:red, 0; green, 0; blue, 0 }  ][line width=0.08]  [draw opacity=0] (8.93,-4.29) -- (0,0) -- (8.93,4.29) -- cycle    ;

\draw (98,172.4) node [anchor=north west][inner sep=0.75pt]    {$p_{1}$};
\draw (119,170) node [anchor=north west][inner sep=0.75pt]   [align=left] {[$A$]};
\draw (159,130.4) node [anchor=north west][inner sep=0.75pt]    {$1'( x)$};
\draw (248,130.4) node [anchor=north west][inner sep=0.75pt]    {$1'( x)$};
\draw (306,172.4) node [anchor=north west][inner sep=0.75pt]    {$p_{2}$};
\draw (328,170) node [anchor=north west][inner sep=0.75pt]   [align=left] {[$A$]};
\draw (217,112.4) node [anchor=north west][inner sep=0.75pt]    {$t_{1}$};
\draw (206,174.4) node [anchor=north west][inner sep=0.75pt]    {$x< 3$};
\draw (57,144.4) node [anchor=north west][inner sep=0.75pt]    {$A.all$};
\draw (367,130.4) node [anchor=north west][inner sep=0.75pt]    {$1'( x)$};
\draw (423,112.4) node [anchor=north west][inner sep=0.75pt]    {$t_{2}$};
\draw (468,121) node [anchor=north west][inner sep=0.75pt]   [align=left] {\textbf{Declarations:}\\Color set $\displaystyle A=([ 1,5])$\\variable $\displaystyle x:A$};
\draw (412,174.4) node [anchor=north west][inner sep=0.75pt]    {$x\leq 1$};

\end{tikzpicture}
    \end{adjustbox}
   \caption{Example CPN}
   \label{fig:UnstablePartition}
\end{subfigure}
\begin{subfigure}[b]{\textwidth}
    \centering
    \begin{adjustbox}{width=0.7\linewidth} 
    \tikzset{every picture/.style={line width=0.75pt}} 

\begin{tikzpicture}[x=0.75pt,y=0.75pt,yscale=-1,xscale=1]

\draw   (120,131.75) .. controls (120,121.95) and (128.06,114) .. (138,114) .. controls (147.94,114) and (156,121.95) .. (156,131.75) .. controls (156,141.55) and (147.94,149.5) .. (138,149.5) .. controls (128.06,149.5) and (120,141.55) .. (120,131.75) -- cycle ;
\draw    (156,131.75) -- (235,131.51) ;
\draw [shift={(238,131.5)}, rotate = 539.8299999999999] [fill={rgb, 255:red, 0; green, 0; blue, 0 }  ][line width=0.08]  [draw opacity=0] (8.93,-4.29) -- (0,0) -- (8.93,4.29) -- cycle    ;
\draw  [fill={rgb, 255:red, 0; green, 0; blue, 0 }  ,fill opacity=1 ] (239,115) -- (247,115) -- (247,148.5) -- (239,148.5) -- cycle ;
\draw   (329,130.75) .. controls (329,120.95) and (337.06,113) .. (347,113) .. controls (356.94,113) and (365,120.95) .. (365,130.75) .. controls (365,140.55) and (356.94,148.5) .. (347,148.5) .. controls (337.06,148.5) and (329,140.55) .. (329,130.75) -- cycle ;
\draw    (248,130.75) -- (326,130.75) ;
\draw [shift={(329,130.75)}, rotate = 180] [fill={rgb, 255:red, 0; green, 0; blue, 0 }  ][line width=0.08]  [draw opacity=0] (8.93,-4.29) -- (0,0) -- (8.93,4.29) -- cycle    ;
\draw  [fill={rgb, 255:red, 0; green, 0; blue, 0 }  ,fill opacity=1 ] (447,115) -- (455,115) -- (455,148.5) -- (447,148.5) -- cycle ;
\draw    (366,130.75) -- (444,130.75) ;
\draw [shift={(447,130.75)}, rotate = 180] [fill={rgb, 255:red, 0; green, 0; blue, 0 }  ][line width=0.08]  [draw opacity=0] (8.93,-4.29) -- (0,0) -- (8.93,4.29) -- cycle    ;

\draw (118,152.4) node [anchor=north west][inner sep=0.75pt]    {$p_{1}$};
\draw (139,150) node [anchor=north west][inner sep=0.75pt]   [align=left] {[\textit{A}]};
\draw (179,110.4) node [anchor=north west][inner sep=0.75pt]    {$1'( x)$};
\draw (268,110.4) node [anchor=north west][inner sep=0.75pt]    {$1'( x)$};
\draw (326,152.4) node [anchor=north west][inner sep=0.75pt]    {$p_{2}$};
\draw (348,150) node [anchor=north west][inner sep=0.75pt]   [align=left] {[\textit{A}]};
\draw (237,92.4) node [anchor=north west][inner sep=0.75pt]    {$t_{1}$};
\draw (226,154.4) node [anchor=north west][inner sep=0.75pt]    {$x < 3$};
\draw (77,114.4) node [anchor=north west][inner sep=0.75pt]    {$ \begin{array}{l}
1'( 1)\\
1'( 2)\\
3'( 3)
\end{array}$};
\draw (387,110.4) node [anchor=north west][inner sep=0.75pt]    {$1'( x)$};
\draw (443,92.4) node [anchor=north west][inner sep=0.75pt]    {$t_{2}$};
\draw (488,101) node [anchor=north west][inner sep=0.75pt]   [align=left] {\textbf{Declarations:}\\Color set $\displaystyle A=([ 1,3])$\\variable $\displaystyle x:A$\\\\Color 1 represents $\displaystyle ([ 1])$\\Color 2 represents $\displaystyle ([ 2])$\\Color 3 represents $\displaystyle ([ 3,5])$};
\draw (432,154.4) node [anchor=north west][inner sep=0.75pt]    {$x\leq 1$};

\end{tikzpicture}
    \end{adjustbox}
   \caption{Quotiented net from Figure~\ref{fig:UnstablePartition}}
   \label{fig:partitionUnfolding}
\end{subfigure}\vspace{2mm}
\begin{subfigure}[b]{\textwidth}
    \centering
    \begin{adjustbox}{width=0.50\linewidth} 
    \tikzset{every picture/.style={line width=0.75pt}} 

\begin{tikzpicture}[x=0.75pt,y=0.75pt,yscale=-1,xscale=1]

\draw (202,132.4) node [anchor=north west][inner sep=0.75pt]  [font=\normalsize]  {$M_{1} =p_{1} :1'( 1)$};
\draw (326,128.4) node [anchor=north west][inner sep=0.75pt]  [font=\Large]  {$\xrightarrow{t_{1}}$};
\draw (202,192.4) node [anchor=north west][inner sep=0.75pt]  [font=\normalsize]  {$M_{2} =p_{1} :1'( 2)$};
\draw (374,134.4) node [anchor=north west][inner sep=0.75pt]  [font=\normalsize]  {$M_{1} '=p_{2} :1'( 1)$};
\draw (374,192.4) node [anchor=north west][inner sep=0.75pt]  [font=\normalsize]  {$M_{2} '=p_{2} :1'( 2)$};
\draw (326,184.4) node [anchor=north west][inner sep=0.75pt]  [font=\Large]  {$\xrightarrow{t_{1}}$};
\draw (150,163.4) node [anchor=north west][inner sep=0.75pt]    {$M_{1} \dequiv M_{2}$};
\draw (473,164.4) node [anchor=north west][inner sep=0.75pt]    {$M_{1} ' \centernot{\dequiv} M_{2} '$};
\draw (191,93.4) node [anchor=north west][inner sep=0.75pt]  [font=\normalsize]  {$\delta ( p_{1}) =\{([ 1,2]) ,([3,5])\} ,\delta ( p_{2}) =\{([ 1]) ,([ 2,5])\}$};

\end{tikzpicture}
    \end{adjustbox}
   \caption{Example of an unstable partition $\delta$ and markings showing why it is unstable}
   \label{fig:UnstablePartitionFiring}
\end{subfigure}\vspace{2mm}
\begin{subfigure}[b]{\textwidth}
    \centering
    \begin{adjustbox}{width=0.5\linewidth} 
    \tikzset{every picture/.style={line width=0.75pt}} 

\begin{tikzpicture}[x=0.75pt,y=0.75pt,yscale=-1,xscale=1]

\draw (190,115.4) node [anchor=north west][inner sep=0.75pt]  [font=\normalsize]  {$\delta '( p_{1}) =\{([ 1]) ,([ 2]) ,([ 3,5])\} ,\delta '( p_{2}) =\{([ 1]) ,([ 2,5])\}$};

\end{tikzpicture}
    \end{adjustbox}       %
   \caption{Example of stable partition $\delta'$}
    \label{fig:StablePartition}
\end{subfigure}\vspace*{-2mm}
\caption{Quotienting example}
\end{figure}

We thus introduce \textit{color partition} on places where
all colors with similar behaviour in a given place are grouped into an \textit{equivalence class},
denoted by $\theta$. For the rest of this section, let us assume a fixed CPN $\mathcal{N} = \CPN{}$.
A partition $\delta$ is a function $\delta : P \xrightarrow{} 2^{2^\mathbb{C}} \setminus \emptyset$
that for a place $p$ returns the equivalence classes of $C(p)$ such that $(\Bcup{\theta \in \delta(p)} \theta) = C(p)$ and $\theta_1 \cap \theta_2 = \emptyset$ for all $\theta_1,\theta_2 \in \delta (p)$ where $\theta_1 \neq \theta_2$.

\begin{definition}\label{def:dequiv} 
Given a partition $\delta$ and markings $M$ and $M'$, we write $M(p) \dequiv M'(p)$ for a $p \in P$ iff for all $\theta \in \delta(p)$ it holds that
 $\Bsum{c \in \theta}M(p)(c) = \textstyle \sum_{c \in \theta}M'(p)(c)$.
We write $M \dequiv M'$ iff $M(p) \dequiv M'(p)$ for all $p \in P$.
A partition $\delta$ is \emph{stable} if the relation $\dequiv$ on markings induced by $\delta$
is a bisimulation.
\end{definition}
\eject



Consider the CPN in Figure~\ref{fig:UnstablePartition}.
The partition shown in the Figure~\ref{fig:UnstablePartitionFiring} is not stable
as demonstrated by the transition firing from $M_1$ and $M_2$ to $M_1'$ resp. $M_2'$
where $M_1 \dequiv M_2$ but $M_1' \centernot{\dequiv} M_2'$.
Figure~\ref{fig:StablePartition} shows an example of a stable partition (here
we describe the partition with integer ranges in the same manner as in integer CPNs).

We now show how a CPN can  be quotiented using a stable partition.
First, we define the notion of binding equivalence under a partition.
\begin{definition}\label{Def:PartBindingEquiv} 
Given a partition $\delta$, a transition $t \in T$ and bindings $b,b' \in B(t)$, we write $b \dtequiv b'$ iff for all $p \in \preset{t}$ and for all $ \theta \in \delta(p)$ it holds that
\begin{equation*}
    \Bsum{c \in \theta}W((p,t),b)(c) = \Bsum{c \in \theta}W((p,t),b')(c)
\end{equation*}
and for all $p \in \postset{t}$ and for all $\theta \in \delta(p)$ it holds that
\begin{equation*}
    \Bsum{c \in \theta}W((t,p),b)(c) = \Bsum{c \in \theta}W((t,p),b')(c).
\end{equation*}
\end{definition}

We can now define classes of equivalent bindings given a partition $\delta$
which are bindings that have the same behaviour for a given transition, formally
$B^\delta(t)\! \defeq\!  \{[b]_t \ | \ b\! \in\! B(t) \} \text{ where } [b]_t\! = \{b' \ | \ b' \dtequiv~b\}$.

\smallskip
For a given  stable partition, we now construct a quotiented CPN where the set of colors
are the equivalence classes of the stable partition and
the set of bindings are the equivalence classes of bindings.
As such, we rewrite the arc and guard evaluation functions to instead consider an equivalence class
of bindings, which is possible since each binding in the equivalence class behaves equivalently.

\begin{definition} \label{def:partitionedCPN} 
Let $\mathcal{N} = \CPN{}$ be a CPN and $\delta$ a stable partition of $\mathcal{N}$.
The quotiented CPN $\mathcal{N}^\delta = \deltasCPN$ is defined by
\begin{enumerate}
    \item $\mathbb{C}^\delta = \Bcup{p \in P} \delta(p)$
    \item $\mathbb{B}^\delta = \biguplus_{t \in T} B^\delta(t)$.
    \item \label{item:guards}
    $G^\delta(t, [b]_t) = G(t, b)$ for all $t \in T$ and $[b]_t \in B(t)$,
    \item $C^\delta (p) = \delta (p)$ for all $p \in P$,
    \item \label{item:arcs}
    $W^\delta((p,t),[b]_t) = S \text{ where } S(\theta) = \Bsum{c \in \theta}W((p,t),b)(c) \text{ for all } \theta \in \delta(p)$ and\\
    $W^\delta((t,p),[b]_t) = S \text{ where } S(\theta) = \Bsum{c \in \theta}W((t,p),b)(c) \text{ for all } \theta \in \delta(p)$ \\ for all $p \in P$, $t \in T$ and $[b]_t \in \mathbb{B}^\delta$,
    \item \label{item:inhibarcs}
    $W_I^\delta(p,t)$ = $W_I(p,t)$ \textit{ for all $p \in P$ and $t \in T$,  and }
    \item \label{item:initialMarking} $\Mdelta_0 (p) =  S \text{ where } S(\theta) = \Bsum{c \in \theta}M_0(p)(c)$ for all $p \in P$ and $\theta \in \delta (p)$.
\end{enumerate}
\end{definition}

We can now present our main correctness theorem, stating that the original and quotiented colored nets
are bisimilar.

\begin{theorem}
Let $\mathcal{N} = \CPN{}$ be a CPN, $\delta$ a stable partition and $\mathcal{N}^\delta = \CPN{\delta}$
the quotiented CPN. Then $M_0 \sim \Mdelta_0$.
\end{theorem}
\begin{proof}
We show that
$R = \{(M,\Mdelta) \ | \ \Bsum{c \in \theta}M(p)(c) = \Mdelta(p)(\theta) \text{ for all } p \in P \text{ and all } \theta \in \delta (p) \}$ is a bisimulation relation.
We first notice that $(M_0,\Mdelta_0) \in R$ by Item \ref{item:initialMarking} in Definition \ref{def:partitionedCPN}.

\medskip
Assume that $(M,\Mdelta) \in R$ and $t \in T$ such that $M \xrightarrow{t} M'$ under binding $b \in B(t)$, we want to show that $\Mdelta \xrightarrow{t} \Mdeltaprime$ under binding $[b]_t \in B(t)$ such that $(M',\Mdeltaprime) \in R$. As such, we need to prove the following:
\begin{align*}
    \bm{(a)} \quad &W^\delta((p,t),[b]_t) \subseteq \Mdelta(p) \text{ for all } p \in P \\
    \bm{(b)} \quad &W_I^\delta(p,t) > |\Mdelta (p)| \text{ for all } p \in P\\
    \bm{(c)} \quad &(M',\Mdeltaprime) \in R \text{ where } \Mdeltaprime(p) = (\Mdelta (p) \setminus W^\delta((p,t),[b]_t)) \uplus W^\delta((t,p),[b]_t)\\ &\text{for all } p \in P \ .
\end{align*}

$\bm{(a)}$ We start by showing $W^\delta((p,t),[b]_t) \subseteq \Mdelta(p)$ for all $p \in P$. First,
because $(M,\Mdelta) \in R$, we know that
\begin{equation} \label{eq:markingsInRelation}
  \Bsum{c \in \theta}M(p)(c) = \Mdelta(p)(\theta) \text{ for all } p \in P \text{ and all } \theta \in \delta (p).
\end{equation}
Since $W((p,t),b) \subseteq M(p)$ we know for all $c \in C(p)$ that $ W((p,t),b))(c) \leq M(p)(c)$
which implies that
\begin{equation} \label{eq:inputArcsTheorem}
   \Bsum{c \in \theta}W((p,t),b))(c) \leq \Bsum{c \in \theta}M(p)(c)
\end{equation}
for all $\theta \in \delta (p)$.
We now show that $W^\delta((p,t),[b]_t) \subseteq \Mdelta(p)$ for all $p \in P$
i.e. $W^\delta((p,t),[b]_t)(\theta) \leq \Mdelta(p)(\theta)$ for all $\theta \in \delta (p)$:

\begin{center}
\begin{tabularx}{\textwidth}{lXr}
    $W^\delta((p,t),[b]_t)(\theta)$ & = & substitute by Def. \ref{def:partitionedCPN} Item \ref{item:arcs} \\
     $\Bsum{c \in \theta}W((p,t),b))(c)$& $\leq$ &by Equation (\ref{eq:inputArcsTheorem}) \\
     $\Bsum{c \in \theta}M(p)(c)$ & = & by Equation (\ref{eq:markingsInRelation}) \\
     $\Mdelta(p)(\theta)$ \ . & &  \\
\end{tabularx}
\end{center}

$\bm{(b)}$ Next we show $W_I^\delta((p,t) > |\Mdelta(p)|$. We know that
\begin{equation} \label{Eq:inhibArcs}
    W_I(p,t) > |M(p)|
\end{equation}
by definition of CPN semantics since $M \xrightarrow{t} M'$. We then observe that for all $p \in P$, it holds
\begin{center}
\begin{tabularx}{\textwidth}{lXr}
    $W^\delta_I(p,t)$ & = &substitute by Def. \ref{def:partitionedCPN} Item \ref{item:inhibarcs}  \\
     $W_I(p,t)$& > &by Equation (\ref{Eq:inhibArcs})\\
     $|M(p)|$ & = & multiset definition\\
     $\Bsum{c \in C(p)}M(p)(c)$ &= & since $(\Bcup{\theta \in \delta(p)} \theta) = C(p)$\\
     $\Bsum{\theta \in \delta (p)}\Bsum{c \in \theta}M(p)(c)$ & = & by Equation (\ref{eq:markingsInRelation})  \\
     $\Bsum{\theta \in \delta(p)}\Mdelta(p)(\theta)$ & = & multiset definition\\
     $|\Mdelta(p)|$ \ .& &\\
\end{tabularx}
\end{center}

$\bm{(c)}$ Lastly, we show that $(M',\Mdeltaprime) \in R$. Assume $p \in P$, $b \in B(t)$ and equivalence class $[b]_t$. We know that $M'(p) = (M(p) \setminus W((p,t),b)) \uplus W((t,p),b)$ and $\Mdeltaprime(p) = (\Mdelta(p) \setminus W^\delta((p,t),[b]_t)) \uplus W^\delta((t,p),[b]_t)$ and we need to show that $\Bsum{c \in \theta}M'(p)(c) = \Mdeltaprime(p)(\theta)$ for all $\theta \in \delta(p)$:

\begin{center}
\begin{tabularx}{\textwidth}{lXr}
    $\Bsum{c\in \theta} M'(p)(c)$ & = & by def. of CPN semantics \\
     $\Bsum{c \in \theta} (M(p) \setminus W((p,t),b) \uplus W((t,p),b))(c)$ & = & \specialcell{substitute by multiset definitions\\ and by enabledness of $t$}\\
     \specialcell{$\Bsum{c \in \theta}M(p)(c) - \Bsum{c \in \theta}W((p,t),b)(c)$\\$ + \Bsum{c \in \theta}W((t,p),b)(c)$} & = &substitute by Def. \ref{def:partitionedCPN} Item \ref{item:arcs}\\
     \specialcell{$\Bsum{c \in \theta}M(p)(c) - W^\delta((p,t),[b]_t)(\theta)$\\ $+ W^\delta((t,p),[b]_t)(\theta) $} &= &since $(M,\Mdelta) \in R$\\
     \specialcell{$\Mdelta(p)(\theta) - W^\delta((p,t),[b]_t)(\theta)$ \\$+ W^\delta((t,p),[b]_t)(\theta) $} & = & by definition of CPN semantics \\
     $\Mdeltaprime(p)(\theta)$ \ . & &
\end{tabularx}
\end{center}

We then have to show that the same is the case for the opposite direction such that assume $(M,\Mdelta) \in R$ and $t \in T$ such that $\Mdelta \xrightarrow{t} \Mdeltaprime$, we want to show that $M \xrightarrow{t} M'$
for some $b \in B(t)$ such that $(M',\Mdeltaprime) \in R$. As such, we want to show that:
\begin{align*}
    \bm{(d)} \quad &W((p,t),b) \subseteq M(p) \text{ for all } p \in P \\
    \bm{(e)} \quad &W_I(p,t) > |M(p)| \text{ for all } p \in P\\
    \bm{(f)} \quad &(M',\Mdeltaprime) \in R \text{ where } M'(p) = (M(p) \setminus W((p,t),b)) \uplus W((t,p),b) \\ &\text{for all } p \in P \ .
\end{align*}

We first notice that $\bm{(e)}$ and $\bm{(f)}$ can be showed by the same argumentation as $\bm{(b)}$ and $\bm{(c)}$.
For the case $\bm{(d)}$, we show that $W((p,t),b) \subseteq M(p) \text{ for all } p \in P$. From $\bm{(a)}$
we know that $W^\delta((p,t),[b]_t) \subseteq \Mdelta(p)$ which implies $\Bsum{c \in \theta}W((p,t),b))(c) \leq \Bsum{c \in \theta}M(p)(c)$ for all $\theta \in \delta(p)$ and $p \in P$.

\medskip
Hence observe that there exists a marking $M_1$ such that $\Bsum{c \in \theta}M_1(p)(c) = M^\delta(p)(\theta)$ and
where $\Bsum{c \in \theta}W((p,t),b)(c) \leq \Bsum{c \in \theta}M_1(p)(c)$ for all $\theta \in \delta(p)$ and $p \in P$. Clearly $t$ is enabled in $M_1$ since $W((p,t),b) \subseteq M_1(p)$ for all $p \in P$ by the multiset definition of $\subseteq$ and we know that the inhibitor arcs do not inhibit the transition by $\bm{(e)}$.

\medskip
We then want to show that $M_1 \dequiv M$, i.e. $\Bsum{c \in \theta}M_1(p)(c) = \Bsum{c \in \theta}M(p)(c)$ for all $\theta \in \delta(p)$ and $p \in P$. Since $(M,\Mdelta) \in R$ we know that $\Bsum{c \in \theta}M(p)(c) = \Mdelta(p)(\theta) = \Bsum{c \in \theta}M_1(p)(c)$ for all $\theta \in \delta(p)$ and $p \in P$ and thus $M_1 \dequiv M$. Since $\delta$ is stable we know that $t$ is enabled in $M$.

\medskip
Thus we know that the opposite direction also holds meaning that $R$ is a bisimulation.
\end{proof}

\subsection{Computing stable partitions}

Our main challenge is how to efficiently compute a stable partition in order
to apply the quotienting technique.
To do so, we first define a partition refinement. 

\begin{definition} 
Given two partitions $\delta$ and $\delta'$ we write
$ \delta \geq \delta'$ if for all  $p \in P$ and all $\theta' \in \delta'(p)$ there exists
$\theta \in \delta (p)$ such that $\theta' \subseteq \theta$.
Additionally, we write $\delta > \delta'$ if $\delta \geq \delta'$ and $\delta' \neq \delta$.
\end{definition}

Note that for any finite CPN as assumed in Remark~\ref{remark:Finite},
 the refinement relation $>$
is well-founded as for any $\delta > \delta'$ the partition $\delta'$ has
strictly more equivalence classes for at least one place $p \in P$. 
We now define also the union of two partitions as the smallest partition that has both
of the partitions as refinements.

\begin{definition} 
Given two partitions $\delta_1$, $\delta_2$ and $p \in P$, let $\xleftrightarrow{}$
be a relation over $\delta_1(p) \cup \delta_2(p)$ such that
$ \theta \xleftrightarrow{} \theta' \textit{ iff } \theta \cap \theta' \neq \emptyset$  where $\theta, \theta' \in \delta_1(p) \cup \delta_2(p)$.
Let $\xleftrightarrow{}^*$ be the reflexive, transitive closure of $\xleftrightarrow{}$
and let $[\theta] \defeq \Bcup{\theta' \in \delta_1(p) \cup \delta_2(p), \theta \xleftrightarrow{}^* \theta'} \theta'$ where $\theta \in \delta_1(p) \cup \delta_2(p)$. Finally, we define
the partition union operator $\sqcup$
by
$(\delta_1 \sqcup \delta_2)(p) = \Bcup{\theta \in \delta_1(p) \cup \delta_2(p)} \{[\theta]\} \text{ for all } p \in P$.
\end{definition}

For example, assume some place $p$ such that $C(p) = \{([1,5])\}$ and partitions $\delta_1$ and $\delta_2$ such that $\delta_1(p) = \{([1,2]),([3,4]),([5])\}$ and $\delta_2(p) = \{([1]),([2,3]),([4]),([5])\}$ then $(\delta_1 \sqcup \delta_2)(p) = \{([1,4]),([5])\}$.

\begin{lemma} \label{lemma:SmallestUnion} \label{lemma:PartitionUnion}  
Let $\delta_1$ and $\delta_2$ be two partitions.
Then
(i) $\delta_1 \sqcup \delta_2 \geq \delta_1$ and $\delta_1 \sqcup \delta_2 \geq \delta_2$, and
(ii) if $\delta_1$ and $\delta_2$ are stable partitions then so is $\delta_1 \sqcup \delta_2$.
\end{lemma}

\begin{proof}
For the first part of the claim,
from definition of partition union, we see that $[\theta]$ is the union of any $\theta'$ that overlaps with $\theta$ and $\delta_1 \sqcup \delta_2$ just collects all such unions for every $\theta$. As such, it is trivial that for any $\theta \in \delta_1 (p)$ there exists $\theta' \in \delta_1 (p) \sqcup \delta_2 (p)$ such that $\theta \subseteq \theta'$ for all $p \in P$, in other words $ \delta_1 \sqcup \delta_2 \geq \delta_1$. The same is the case for $\delta_2$.

\medskip
For the second part of the claim,
let $\delta = \delta_1 \sqcup \delta_2$. Assume $M$ and $M'$ such that $M \dequiv M'$ i.e. for all $p \in P$ and all $[\theta ]\in \delta (p)$ it holds that $\Bsum{c \in [\theta]} M(p)(c) = \Bsum{c \in [\theta]} M'(p)(c)$.
We want to show that $\delta$ is stable, i.e that $\dequivparam{}$ is a bisimulation relation,
while assuming that $\dequivparam{1}$ and $\dequivparam{2}$ are bisimulation relations.
Assume a fixed $p \in P$.
From definition of partition union, we get that for every fixed $[\theta] \in \delta (p)$
and for all $c,c' \in [\theta]$ there exist $c_1,...,c_k \in [\theta]$ such that
$c,c_1 \in \theta_1$, $c_1,c_2 \in \theta_2$, $c_2,c_3 \in \theta_3$, \ldots,
$c_k,c' \in \theta_k$, where $\theta_i \in \delta_1 (p) \cup \delta_2 (p)$
for all $i, 1 \leq i \leq k$.
Let us define markings $M_i$, for $1 \leq i \leq k$,
such that $M_i(p)(c_i) = \Bsum{c \in [\theta]} M(p)(c)$ and
$M_i(p)(c) = 0$ if $c \neq c_i$ and $c \in [\theta]$, otherwise
$M_i(p)(c) = M(p)(c)$. In other words, in order to obtain $M_i$, we replace in $M$
all colors in the equivalence class $[\theta]$ with the color $c_i$ and obtain
the chain
$M(p) \dequivparam{j_1} M_1(p) \dequivparam{j_2} M_2(p) \dequivparam{j_3} ... \dequivparam{j_k} M_k(p)$
where $j_i \in \{1,2\}$, i.e. $M_i$ and $M_{i+1}$ are related either by
$\dequivparam{1}$ or $\dequivparam{2}$, both of them being bisimulation relations. We
can also observe that $M(p) \dequiv M_k(p)$.
We can then repeat this process for
all other equivalence classes in $\delta(p)$ in order to conclude that
$M(p) \dequiv M'(p)$.
The same process can then be applied to all $p \in P$, implying that
$M \dequiv M'$. 
Hence there is a chain of markings as above starting from $M$ and ending in $M'$.
Since $\delta_1$ and $\delta_2$ are stable, meaning that both $\dequivparam{1}$ and $\dequivparam{2}$
are bisimulation relations, this implies that every transition from $M$ can be
(by transitivity) matched by a transition from $M'$ and vice versa,
implying that $\dequivparam{}$ is a bisimulation relation
and $\delta$ is so stable.
\end{proof}

The lemma above implies the existence of a unique maximum stable partition.

\begin{theorem} 
There is a unique maximum stable partition $\delta$ 
such that $\delta \geq \delta'$
for all stable partitions~$\delta'$. 
\end{theorem}

\begin{proof}
We prove this by contradiction. Assume two maximum stable partitions $\delta_1$ and $\delta_2$ where $\delta_1 \neq \delta_2$.
By Lemma \ref{lemma:PartitionUnion} (ii) we know that $\delta_1 \sqcup \delta_2$ is stable and by Lemma \ref{lemma:SmallestUnion} (i) we know that $\delta_1 \sqcup \delta_2 \geq \delta_1$ and $\delta_1 \sqcup \delta_2 \geq \delta_2$. Thus $\delta_1$ and $\delta_2$ cannot both be maximum stable partitions. 
\end{proof}

In order to provide an algorithm for computing a stable partition, we
define the maximum arc size for a given CPN $\mathcal{N}$ as the function
$\textit{max}(\mathcal{N}) = \max_{p \in P, t\in T,b \in \mathbb{B}}(|W((p,t),b)|, |W((t,p),b)|)$.
The set of all markings smaller than the \textit{max} arc size over $\mathcal{N}$ is defined by
$\mathbb{M}^{bounded}(\mathcal{N}) = \{ M \in \mathbb{M}(\mathcal{N}) \ | \ |M(p)| \leq \textit{max}(\mathcal{N}) \text{ for all } p \in P \}$.
As such, $\mathbb{M}^{bounded}(\mathcal{N})$ is a finite set
of all bounded markings of $\mathcal{N}$ with cardinality less than or equal
to $\textit{max}(\mathcal{N})$.

\medskip
\hspace*{-1mm}In order to compute stable partitions we need to show some properties for markings in $\mathbb{M}^{bounded}(\mathcal{N})$. The following lemma shows that if there are two non-bisimilar markings that break
the fact that the relation $\dequiv$ is a bisimulation, then there are also two markings
that are bounded and that also break the bisimulation property.

\begin{lemma}\label{lemma:BoundedMarkings}
Let $\mathcal{N}$ be a CPN and $\delta$ a partition. Then for all $t \in T$ it holds that
\begin{enumerate}
    \item[$\bm{(a)}$] if there exist $M_1,M_2 \in \mathbb{M}(\mathcal{N})$ such that $M_1 \dequiv M_2$, $M_1 \xrightarrow{t} \text{ and } \ M_2 \centernot{\xrightarrow{t}}$ then there exist $M_3,M_4 \in \mathbb{M}^{bounded}$ such that $M_3 \dequiv M_4$, $M_3 \xrightarrow{t} \text{ and } \ M_4 \centernot{\xrightarrow{t}}$, and
    \item[$\bm{(b)}$] if there exist $M_1,M_2 \in \mathbb{M}(\mathcal{N})$ where $M_1 \dequiv M_2$ and there exists $M_1' \in \mathbb{M}(\mathcal{N})$ such that $M_1 \xrightarrow{t} M_1'$ and for all $M_2' \in \mathbb{M}(\mathcal{N})$ where $M_2 \xrightarrow{t} M_2'$ it holds that $M_1' \centernot{\dequiv} M_2'$ then there exists $M_3,M_4 \in \mathbb{M}^{bounded}(\mathcal{N})$ where $M_3 \dequiv M_4$ and there exists $M_3' \in \mathbb{M}^{bounded}(\mathcal{N})$ such that $M_3 \xrightarrow{t} M_3'$ and for all $M_4' \in \mathbb{M}^{bounded}(\mathcal{N})$ where $M_4 \xrightarrow{t} M_4'$ it holds that $M_3' \centernot{\dequiv} M_4'$.
\end{enumerate}
\end{lemma}
\begin{proof}
Recall that $\textit{max}(\mathcal{N})$ is defined as the largest cardinality of all arc multisets in $\mathcal{N}$, i.e. $|W((p,t),b)| \leq \textit{max}(\mathcal{N})$ for all $p \in P$ and $b \in B(t)$.
\begin{enumerate}
    \item[$\bm{(a)}$] Let $M_1,M_2 \in \mathbb{M}(\mathcal{N})$ such that $M_1 \dequiv M_2$, $M_1 \xrightarrow{t} \text{ and } \ M_2 \centernot{\xrightarrow{t}}$. We construct a marking $M_3$ such that $M_3(p) = W((p,t),b)$ for all $p \in P$ and some $b \in B(t)$.  It clearly follows that $|M_3(p)| \leq \textit{max}(\mathcal{N})$ for all
$p \in P$, hence $M_3 \in \mathbb{M}^{bounded}(\mathcal{N})$. We see that $M_1(p) = M_3(p) \uplus \overline{M}_3(p)$ since $M_3(p) \subseteq M_1(p)$ for all $p \in P$ where $\overline{M}_3(p)$ describes the remaining tokens in $M_1(p)$ that are not in $M_3(p)$.
    We know that no inhibitor arc can be the reason that $M_2$ is not enabled because $M_1 \dequiv M_2$. We also know that $M_2(p) \centernot{\subseteq} W((p,t),b)$ for at least one $p \in P$ for all $b \in B(t)$.

    We pick a marking $M_4$ where $M_4 \dequiv M_3$, $M_4(p) \centernot{\subseteq} W((p,t),b)$ for at least one $p \in P$ and $M_2(p) = M_4(p) \uplus \overline{M}_4(p)$ for all $p \in P$ such that $\overline{M}_4 \dequiv \overline{M}_3$. Notice that $M_4 \in \mathbb{M}^{bounded}(\mathcal{N})$. We know that $M_4$ exists because $M_2 \dequiv M_1$ and $M_2(p) \centernot{\subseteq} W((p,t),b)$ meaning that $M_4(p) \uplus \overline{M}_4(p) \centernot{\subseteq} W((p,t),b)$ and thus $M_4(p) \centernot{\subseteq} W((p,t),b)$ for some $p \in P$.
    \item[$\bm{(b)}$] Let $M_1,M_2 \in \mathbb{M}(\mathcal{N})$ such that $M_1 \dequiv M_2$ and $M_1' \in \mathbb{M}(\mathcal{N})$ such that $M_1 \xrightarrow{t} M_1'$ while for all $M_2' \in \mathbb{M}(\mathcal{N})$ where $M_2 \xrightarrow{t} M_2'$ we know that $M_1' \centernot{\dequiv} M_2'$. We construct a marking $M_3$ exactly as before such that $M_3(p) = W((p,t),b)$ for all $p \in P$ and some $b \in B(t)$ and $M_1(p) = M_3(p) \uplus \overline{M}_3(p)$ for all $p \in P$.

    We then pick a marking $M_4$ such that $M_4 \dequiv M_3$
and $M_2(p) = M_4(p) \uplus \overline{M}_4(p)$ for all $p \in P$ and $b \in B(t)$ such that $\overline{M}_4 \dequiv \overline{M}_3$. We know that $M_4 \in \mathbb{M}^{bounded}(\mathcal{N})$ since $M_4 \dequiv M_3$.
    Let $M_3' \in \mathbb{M}^{bounded}(\mathcal{N})$ such that $M_3 \xrightarrow{t} M_3'$, which is possible because $M_3(p) \subseteq M_1(p)$ for all $p \in P$ such that no inhibitor arc can inhibit $M_3$. For the sake of contradiction now assume there exists a marking $M_4' \in \mathbb{M}^{bounded}(\mathcal{N})$ such that $M_4 \xrightarrow{t} M_4'$ and $M_3' \dequiv M_4'$. Then notice that we can let $M_1'(p) = M_3'(p) \uplus \overline{M}_3(p)$ where $M_1 \xrightarrow{t} M_1'$ since $M_3(p) \subseteq M_1(p)$ for all $p \in P$ and let $M_2'(p) = M_4'(p) \uplus \overline{M_4}(p)$ where $M_2 \xrightarrow{t} M_2'$ since $M_4(p) \subseteq M_2(p)$ for all $p \in P$. But since $M_3'(p) \dequiv M_4'(p)$ and $\overline{M}_3(p) \dequiv \overline{M}_4(p)$ it means that $M_3'(p) \uplus \overline{M}_3(p) \dequiv M_4'(p) \uplus \overline{M_4}(p)$ for all $p \in P$, i.e. $M_1' \dequiv M_2'$. However, this contradicts the conditions of $\bm{(b)}$, and as such $M_3' \dequiv M_4'$ cannot hold.
\end{enumerate}

\vspace*{-7mm}
\end{proof}

\begin{algorithm}[!t]
\SetAlgoLined
\textbf{Input}: $\mathcal{N} = \CPN{}$\\
\textbf{Output}: Stable partition $\delta$\\
let $\delta(p) := \{C(p)\}$ for all $p \in P$\\
\For{$t \in T$}{ 
    \While{$\exists M_1,M_2 \in \mathbb{M}^{bounded}(\mathcal{N}).M_1 \dequiv M_2 \wedge M_1 \centernot{\xrightarrow{t}} \wedge M_2 \xrightarrow{t} $}
        {
            pick $\delta' < \delta$ such that $M_1 \centernot{\dequivprime} M_2$ \\   
           $\delta := \delta'$\\
        }
}
let $\mathcal{Q} := P$ \ //Waiting list of places\\
\While{$\mathcal{Q} \neq \emptyset$}{
    let $p \in \mathcal{Q}$; $\mathcal{Q} := \mathcal{Q} \setminus \{p\}$\\

    \For{$t \in \preset{p}$}{
        \If{$\exists M_1, M_2 \in \mathbb{M}^{bounded}(\mathcal{N}).M_1 \dequiv M_2.\exists M_1' \in \mathbb{M}^{bounded}(\mathcal{N}). M_1 \xrightarrow{t} M_1' \wedge \forall M_2' \in \mathbb{M}^{bounded}(\mathcal{N}).M_2 \xrightarrow{t} M_2' \wedge M_1'(p) \centernot{\dequiv} M_2'(p)$}
        {
            pick $\delta' < \delta$ such that $M_1\centernot{\dequivprime} M_2$ 
           and $\delta'(p') = \delta(p')$ for all $p' \in P \setminus \preset{t}$ \\
            $\mathcal{Q} := \mathcal{Q} \cup \{p' \ | \ \delta'(p') \neq \delta(p')\}$ \\
            $\delta := \delta'$
        }

    }
}
\Return $\delta$
\caption{$\stabilize(\mathcal{N})$}
\label{Algo:StablePartition}
\end{algorithm}

Algorithm \ref{Algo:StablePartition} now gives a procedure for computing a stable partition
over a given CPN. It starts with an initial partition where every color in the color domain
is in the same equivalence class for each place. The algorithm is then split into two parts.
The first part from line 4 to 9 creates an initial partition applying the guard restrictions
to the input places of the transitions. The second part from line 11 to 20 back propagates
the guard restrictions throughout the net such that only colors that behave the same
are quotiented together. Depending on the choices in lines 6 and 15, the algorithm may
return the maximum stable partition, however in the practical implementation this is not
guaranteed due to an approximation of the guard/arc expression analysis.

\begin{theorem}
Given a CPN $\mathcal{N}$, the algorithm $\stabilize(\mathcal{N})$ terminates
and returns a stable partition of $\mathcal{N}$.
\end{theorem}
\begin{proof}
We first prove that $\stabilize (\mathcal{N})$ terminates. Notice that each iteration produces a new $\delta$ according to the $>$ operator, and since the operator is well-founded we know that the algorithm terminates.
We then show that for $\delta = \stabilize (\mathcal{N})$, $\delta$ is a stable partition of $\mathcal{N}$. Recall, a partition $\delta$ is stable iff for any markings $M_1 \dequiv M_2$ whenever $M_1 \xrightarrow{t} M_1'$ for some $t$ and $M_1'$ then $M_2 \xrightarrow{t} M_2'$ for some $M_2'$ such that $M_1' \dequiv M_2'$.
We prove this by contradiction. Assume that $\delta$ is not a stable partition. As such there must exists markings $M_1,M_2 \in \mathbb{M}(\mathcal{N})$ such that $M_1 \dequiv M_2$ and there exists a marking $M_1' \in \mathbb{M}(\mathcal{N})$ such that $M_1 \xrightarrow{t} M_1'$ for some transition $t$ where for all $M_2' \in \mathbb{M}(\mathcal{N})$ such that $M_2 \xrightarrow{t} M_2'$ then $M_1' \centernot{\dequiv} M_2'$.
This is exactly the property stated in the if statement on line 17 and we know from Lemma \ref{lemma:BoundedMarkings} that if the property is satisfied with two markings from $\mathbb{M}(\mathcal{N})$ then there exists two markings from $\mathbb{M}^{bounded}(\mathcal{N})$ that also satisfy the property.
\end{proof}

\subsection{Stable partition Algorithm for integer CPNs}

The $\stabilize$ computation presented in Algorithm \ref{Algo:StablePartition} can be used to
find a stable partition for any finite CPN. However, implementation-wise it is inefficient to
represent every color in a given equivalence class individually.
Hence, for integer CPN we represent an equivalence class as a tuple of ranges.
As an example of computing stable partitions with Algorithm~\ref{Algo:StablePartition},
consider the integer CPN in Figure~\ref{fig:GuardPartition}. Table~\ref{tab:exampleTable} shows
the different stages that $\delta$ undergoes in order to become stable.
In iteration 0, the guard restrictions from the first for-loop are applied,
followed by the iterations of the main while-loop.
In our implementation, we do not iterate through every bounded marking and we instead
(for efficiency reasons)
statically analyze the places, arcs and guards in order to partition the color sets.
For example, in iteration number 1, we consider the place $p_3$ and we can see that the colors in the range $[1,3]$ must
be distinguished from the color $4$. This partitioning propagates back to the place $p_1$ as firing the transition $t_1$ moves
tokens from $p_1$ to $p_3$ without changing its color.

\begin{figure}[!t]
    \centering
    \resizebox{0.75\textwidth}{!}{%
     \tikzset{every picture/.style={line width=0.75pt}} 

\begin{tikzpicture}[x=0.75pt,y=0.75pt,yscale=-1,xscale=1]

\draw   (98.5,79.25) .. controls (98.5,70.28) and (105.78,63) .. (114.75,63) .. controls (123.72,63) and (131,70.28) .. (131,79.25) .. controls (131,88.22) and (123.72,95.5) .. (114.75,95.5) .. controls (105.78,95.5) and (98.5,88.22) .. (98.5,79.25) -- cycle ;
\draw   (327.5,78.25) .. controls (327.5,69.28) and (334.78,62) .. (343.75,62) .. controls (352.72,62) and (360,69.28) .. (360,78.25) .. controls (360,87.22) and (352.72,94.5) .. (343.75,94.5) .. controls (334.78,94.5) and (327.5,87.22) .. (327.5,78.25) -- cycle ;
\draw  [fill={rgb, 255:red, 0; green, 0; blue, 0 }  ,fill opacity=1 ] (226,116) -- (233,116) -- (233,147.5) -- (226,147.5) -- cycle ;
\draw    (130,85.5) -- (223.24,125.32) ;
\draw [shift={(226,126.5)}, rotate = 203.13] [fill={rgb, 255:red, 0; green, 0; blue, 0 }  ][line width=0.08]  [draw opacity=0] (8.93,-4.29) -- (0,0) -- (8.93,4.29) -- cycle    ;
\draw   (99.5,176.25) .. controls (99.5,167.28) and (106.78,160) .. (115.75,160) .. controls (124.72,160) and (132,167.28) .. (132,176.25) .. controls (132,185.22) and (124.72,192.5) .. (115.75,192.5) .. controls (106.78,192.5) and (99.5,185.22) .. (99.5,176.25) -- cycle ;
\draw    (130,169.5) -- (222.17,137.48) ;
\draw [shift={(225,136.5)}, rotate = 520.8399999999999] [fill={rgb, 255:red, 0; green, 0; blue, 0 }  ][line width=0.08]  [draw opacity=0] (8.93,-4.29) -- (0,0) -- (8.93,4.29) -- cycle    ;
\draw    (234,132.5) -- (326.29,88.78) ;
\draw [shift={(329,87.5)}, rotate = 514.65] [fill={rgb, 255:red, 0; green, 0; blue, 0 }  ][line width=0.08]  [draw opacity=0] (8.93,-4.29) -- (0,0) -- (8.93,4.29) -- cycle    ;
\draw   (327.5,178.25) .. controls (327.5,169.28) and (334.78,162) .. (343.75,162) .. controls (352.72,162) and (360,169.28) .. (360,178.25) .. controls (360,187.22) and (352.72,194.5) .. (343.75,194.5) .. controls (334.78,194.5) and (327.5,187.22) .. (327.5,178.25) -- cycle ;
\draw    (234,132.5) -- (326.24,171.34) ;
\draw [shift={(329,172.5)}, rotate = 202.82999999999998] [fill={rgb, 255:red, 0; green, 0; blue, 0 }  ][line width=0.08]  [draw opacity=0] (8.93,-4.29) -- (0,0) -- (8.93,4.29) -- cycle    ;
\draw  [fill={rgb, 255:red, 0; green, 0; blue, 0 }  ,fill opacity=1 ] (226,206) -- (233,206) -- (233,237.5) -- (226,237.5) -- cycle ;
\draw    (132.81,187.56) -- (225,222.5) ;
\draw [shift={(130,186.5)}, rotate = 20.75] [fill={rgb, 255:red, 0; green, 0; blue, 0 }  ][line width=0.08]  [draw opacity=0] (8.93,-4.29) -- (0,0) -- (8.93,4.29) -- cycle    ;
\draw    (236.82,221.49) -- (273.01,208.54) -- (329,188.5) ;
\draw [shift={(234,222.5)}, rotate = 340.31] [fill={rgb, 255:red, 0; green, 0; blue, 0 }  ][line width=0.08]  [draw opacity=0] (8.93,-4.29) -- (0,0) -- (8.93,4.29) -- cycle    ;
\draw  [fill={rgb, 255:red, 0; green, 0; blue, 0 }  ,fill opacity=1 ] (445,63) -- (452,63) -- (452,94.5) -- (445,94.5) -- cycle ;
\draw    (360,78.25) -- (440.67,78.01) ;
\draw [shift={(443.67,78)}, rotate = 539.8299999999999] [fill={rgb, 255:red, 0; green, 0; blue, 0 }  ][line width=0.08]  [draw opacity=0] (8.93,-4.29) -- (0,0) -- (8.93,4.29) -- cycle    ;
\draw  [fill={rgb, 255:red, 0; green, 0; blue, 0 }  ,fill opacity=1 ] (445,163) -- (452,163) -- (452,194.5) -- (445,194.5) -- cycle ;
\draw    (360,178.25) -- (440.67,178.01) ;
\draw [shift={(443.67,178)}, rotate = 539.8299999999999] [fill={rgb, 255:red, 0; green, 0; blue, 0 }  ][line width=0.08]  [draw opacity=0] (8.93,-4.29) -- (0,0) -- (8.93,4.29) -- cycle    ;

\draw (106,100.4) node [anchor=north west][inner sep=0.75pt]    {$p_{1}$};
\draw (222,94.4) node [anchor=north west][inner sep=0.75pt]    {$t_{1}$};
\draw (222,184.4) node [anchor=north west][inner sep=0.75pt]    {$t_{2}$};
\draw (150.67,77.73) node [anchor=north west][inner sep=0.75pt]    {$1'( x)$};
\draw (150.33,131.73) node [anchor=north west][inner sep=0.75pt]    {$1'( y)$};
\draw (264.33,129.73) node [anchor=north west][inner sep=0.75pt]    {$1'( x,y)$};
\draw (257.67,182.73) node [anchor=north west][inner sep=0.75pt]    {$1'( x,y)$};
\draw (160.67,182.73) node [anchor=north west][inner sep=0.75pt]    {$1'( y+1)$};
\draw (441,42.4) node [anchor=north west][inner sep=0.75pt]    {$t_{3}$};
\draw (264.67,83.73) node [anchor=north west][inner sep=0.75pt]    {$2'( x)$};
\draw (381.67,57.73) node [anchor=north west][inner sep=0.75pt]    {$1'( x)$};
\draw (427.67,99.4) node [anchor=north west][inner sep=0.75pt]    {$x\ \leq \ 3$};
\draw (209.33,244.73) node [anchor=north west][inner sep=0.75pt]    {$y\ \geq \ 3$};
\draw (441,142.4) node [anchor=north west][inner sep=0.75pt]    {$t_{4}$};
\draw (375.67,155.73) node [anchor=north west][inner sep=0.75pt]    {$1'( x,y)$};
\draw (428.67,199.4) node [anchor=north west][inner sep=0.75pt]    {$y\ < \ 2$};
\draw (479,90) node [anchor=north west][inner sep=0.75pt]   [align=left] {\textbf{Declarations:}\\Color set $\displaystyle A=([ 1,4])$\\Color set $\displaystyle AA=A\ \times A$\\variable $\displaystyle x:A$\\variable $\displaystyle y:A$};
\draw (66,72) node [anchor=north west][inner sep=0.75pt]   [align=left] {[\textit{A}]};
\draw (68,172) node [anchor=north west][inner sep=0.75pt]   [align=left] {[\textit{A}]};
\draw (331,37) node [anchor=north west][inner sep=0.75pt]   [align=left] {[\textit{A}]};
\draw (328,138) node [anchor=north west][inner sep=0.75pt]   [align=left] {[\textit{AA}]};
\draw (66,95.4) node [anchor=north west][inner sep=0.75pt]    {$1'( 1)$};
\draw (69,195.4) node [anchor=north west][inner sep=0.75pt]    {$1'( 3)$};
\draw (106,196.4) node [anchor=north west][inner sep=0.75pt]    {$p_{2}$};
\draw (336,99.4) node [anchor=north west][inner sep=0.75pt]    {$p_{3}$};
\draw (336,200.4) node [anchor=north west][inner sep=0.75pt]    {$p_{4}$};

\end{tikzpicture}
    }%
    \vspace{-3mm}
    \caption{Example CPN}
    \label{fig:GuardPartition}
    \vspace*{\floatsep}
     \captionof{table}{Stages of $\delta$ throughout Algorithm~\ref{Algo:StablePartition} for CPN in Figure
   \ref{fig:GuardPartition}. The 0'th iteration is the state of $\delta$ just before the while loop begins.
   The symbol '-' indicates that the value is the same as in the previous row.}
     \label{tab:exampleTable}
    \resizebox{\textwidth}{!}{
     
\begin{tabular}{l|c|c|c|c|l}
 Iteration & \multicolumn{1}{c|}{$p_1$} & \multicolumn{1}{c|}{$p_2$} & \multicolumn{1}{c|}{$p_3$} & \multicolumn{1}{c|}{$p_4$}                                                                                               & \multicolumn{1}{c}{$\mathcal{Q}$} \\
\hline
0         & $\{([1,4])\}$            & $\{([1,4])\}$            & $\{([1,3]),([4])\}$    &
    \begin{tabular}[c]{@{}l@{}}$\{([1,4],[1]), ([1,4],[2]), ([1,4],[3,4])\}$\end{tabular} & $\{p_1,p_2,p_3,p_4\}$             \\ \hline
1, $p=p_3$         & $\{([1,3]),([4])\}$            & -            & -    & - & $\{p_1,p_2,p_4\}$             \\ \hline
2, $p=p_4$         & -            & \specialcelltwo{$\{([1]),$ $([2]), ([3,4])\}$}            & -    & - & $\{p_1,p_2\}$             \\ \hline
3, $p=p_2$         & -            & -            & -    & \begin{tabular}[c]{@{}l@{}}$\{([1,4],[1]), ([1,4],[2]),$ \\ $([1,4],[3]), ([1,4],[4])\}$\end{tabular} & $\{p_1,p_4\}$             \\ \hline
4,   $p=p_4$      & -            & \specialcelltwo{$\{([1]),$ $([2]), ([3]), ([4])\}$}            & -    & - & $\{p_1,p_2\}$             \\ \hline
5, $p=p_2$      & -            & -            & -    & - & $\{p_1\}$             \\ \hline
6, $p=p_1$      & -            & -            & -    & - & $\{\}$             \\ 
\hline
\end{tabular}

     }\vspace*{7mm}
       \captionof{table}{Stages of $\alpha$ when computing the fixed point of $E$ for the CPN in
Figure~\ref{fig:GuardPartition}. The symbol '-' indicates that the value is the same as in the previous row.}
       \resizebox{0.75\textwidth}{!}{%
    
\begin{tabular}{l|c|c|c|c}
Iteration & \multicolumn{1}{c|}{$p_1$} & \multicolumn{1}{c|}{$p_2$} & \multicolumn{1}{c|}{$p_3$} & \multicolumn{1}{c}{$p_4$}                                                                                                \\ \hline
0, $\alpha = \alpha_0$         & $\{([1])\}$            & $\{([3])\}$            & $\{\}$    & \begin{tabular}[c]{@{}l@{}}$\{\}$\end{tabular}            \\ \hline
1, $t=t_1$ 
         & -            & -            & $\{([1])\}$   & $\{([1],[3])\}$         \\ \hline
2, $t=t_2$ 
& -            & $\{([3,4])\}$            & -    & -        \\ \hline
3, $t=t_1$ 
& -            & -            & -    & $\{([1],[3,4])\}$           \\ \hline
4,   $t=t_2$ 
& -            & $\{([3,4]),([1])\}$            & -    & -            \\ \hline
5, $t=t_1$ 
     & -            & -            & -    & $\{([1],[3,4]),([1],[1])\}$            \\ \hline
\end{tabular}


    }
    \label{tab:cfpTable}\vspace*{-5mm}
\end{figure}

\section{Color approximation}\label{Sec:ColorApprox}
We now introduce another technique for safely overapproximating what colors can be
present in each place of a CPN. Let
$\mathcal{N} = \CPN{}$ be a fixed CPN for the rest of this section.
A \emph{color approximation} is a function $\alpha : P \xrightarrow{} 2^{\mathbb{C}}$
where $\alpha (p)$ approximates the possible colors in place $p \in P$ such that $\alpha (p) \subseteq C(p)$. Let $\mathbb{A}$ be the set of all color approximations.
For a marking $M$ and color approximation $\alpha$, we write $M \subseteq \alpha$ iff $\colorsimp(M(p)) \subseteq \alpha(p)$ for all $p \in P$.
A \emph{color expansion} is a function $E : \mathbb{A} \xrightarrow{} \mathbb{A}$ defined by
\begin{align*}
    E(\alpha)(p) = \left \{\begin{array}{lll}
         &\alpha(p) \cup \colorsimp(W((t,p),b)) \textit{ }&\text{if } \exists t \in T.\exists b \in B(t).
         \\ &&\colorsimp(W((q,t),b)) \subseteq \alpha(q) \text{ for all } q \in P
         \\
         &\alpha (p) \textit{ }& \text{otherwise.}
    \end{array}\right.
\end{align*}

A color expansion iteratively expands the possible colors that exist in each place but without keeping
a count of how many tokens of each color are present (if a color is present in a place, we assume
that arbitrary many copies of the color are present).
The expansion function obviously preserves the following property.

\begin{lemma}\label{Lemma:Monotonicity} 
Let $\alpha$ be a color approximation then $\alpha(p) \subseteq E(\alpha)(p)$ for all $p \in P$.
\end{lemma}

Let $\alpha_0$ be the initial approximation such that $\alpha_0(p) \defeq \colorsimp(M_0(p))$ for all $p \in P$.
Since $E$ is a monotonic function on a complete lattice, we can compute its
minimum fixed point and formulate the following key lemma. 

\begin{lemma} \label{alphatheorem} 
Let $\alpha$ be a minimum fixed point of $E$ such that $\alpha_0(p) \subseteq \alpha(p)$ for all
$p\in P$.
If $M_0 \xrightarrow{}^* M$ then $M \subseteq \alpha$.
\end{lemma}
\begin{proof}
By induction on $k$ we prove if $M_0 \xrightarrow{}^k M$ then $M \subseteq \alpha$.
\textit{Base step.} If $k=0$, we know that $M_0 \subseteq \alpha$
by the assumption of the lemma.
\textit{Induction step.} Let $M_0 \xrightarrow{}^k M \xrightarrow{t} M'$ by some transition $t$ with some binding $b \in B(t)$. We want to show that $M' \subseteq \alpha$. By induction hypothesis we know that $M \subseteq \alpha$. If $M \xrightarrow{t} M'$ for some $b \in B(t)$, then $M'(p) = (M(p) \setminus W((p,t),b)) \uplus W((t,p),b)$ for all $p \in P$. Since $E(\alpha)$ is a fixed point then $\alpha(p) = \alpha(p) \cup \colorsimp(W((t,p),b))$ for transition $t$ under binding $b$ for all $p \in P$ i.e. $\colorsimp(W((q,t),b)) \subseteq \alpha(q)$ for all $q \in P$. Thus we get $M' \subseteq \alpha$. 
\end{proof}

Given a color approximation $\alpha$ satisfying the preconditions of Lemma~\ref{alphatheorem},
we can now construct a reduced CPN
$\mathcal{N}^\alpha = (P, T, \mathbb{C}, \mathbb{B}, C^\alpha, G, W, W_I, M_0 )$ where $C^\alpha(p) = \alpha(p)$ for all $p \in P$.
The net $\mathcal{N}^\alpha$ can hence have possibly smaller set of colors in its color domains
and it satisfies the following theorem.

\begin{theorem}
The reachable fragments from the initial marking $M_0$ of the LTSs
generated by $\mathcal{N}$ and $\mathcal{N}^\alpha$ are isomorphic.
\end{theorem}
\begin{proof}
By Lemma~\ref{alphatheorem} we know that
$M \subseteq \alpha$ for any marking $M \in \mathbb{M}(\mathcal{N})$ reachable from $M_0$.
Since the reachable fragments of $\mathcal{N}$ and $\mathcal{N}^\alpha$ are exactly the reachable markings, we know that the reachable fragments are isomorphic. 
\end{proof}

\subsection{Computing color approximation on integer CPNs}

As with color quotienting, representing each color individually becomes inefficient.
We thus employ integer ranges to represent color approximations.
Consider the approximation $\alpha$ where $\alpha(p) = \{(1,2),(2,2),(3,2),(5,6),(5,7)\}$
are possible colors (pairs of integers) in the place $p$; this
can be more compactly represented as a set of tuples of ranges $\{([1,3],[2]),([5],[6,7])\}$.

However, computing the minimum fixed point of $E$ using ranges is not as trivial as using complete
color sets. To do so, we need to compute new ranges depending on arcs and guards.
We demonstrate this on the CPN in Figure \ref{fig:GuardPartition}.
Table~\ref{tab:cfpTable} shows the computation of the minimum fixed point of $E$, starting
from the initial approximation $\alpha_0$. 
For example, in iteration number 5, we check if firing transition $t_1$ can produce
any additional tokens to the places $p_3$ and $p_4$. Clearly, there is no change to the
possible token colors in $p_3$ as $\alpha(p_1)$ did not change, however the addition
of the integer range $[1]$ to $\alpha(p_2)$ in the previous iteration now allows us
to produce a new token color $(1,1)$ into $p_4$ and hence we add the singleton range $([1],[1])$
to $\alpha(p_4)$. Due to the guard $y \geq 3$ on the transition $t_2$, we know that
the added token $(1,1)$ does not generate any further behaviour and hence we reached a fixed point.

\section{Experiments}
We implemented the quotienting method from Section~\ref{Sec:Partitioning}
as well as the color approximation method from Section \ref{Sec:ColorApprox} in C++ as
an extension to the verification engine \textit{verifypn}~\cite{verifypnPaper}
from the TAPAAL toolchain~\cite{TAPAALTool}.

\medskip
We perform experiments by comparing several different approaches; the quotienting approach (method A),
the color approximation approach (method B) and the combination of
both (method A+B) against
\begin{itemize}
\itemsep=0.95pt
\item the unfolder MCC~\cite{MCCUnfolder} (used also by TINA~\cite{TINAtool} and LoLA~\cite{LoLA2}),
\item ITS-Tools unfolder~\cite{ITSToolPaper},
\item the Spike unfolder~\cite{SpikeTool} (also used by
MARCIE~\cite{MARCIETool} and Snoopy~\cite{SnoopyTool}),
\item and \textit{verifypn} TAPAAL unfolder with methods A and B disabled, referred to as Tapaal.
\end{itemize}
We compare the tools on the complete set of CPN nets and queries from
2021 Model Checking Contest~\cite{mcc:2021}\footnote{We omit the \emph{UtilityControlRoom} family of models as
they contain operators not supported by ITS-Tools.}.
The experiments are conducted on a compute cluster running Linux version 5.4.0, where each experiment is conducted on a AMD Epyc 7642 processor with a 15 GB memory limit and 5 minute timeout.
To reduce noise in the experiments, we read and write models to a RAM-disk (not included in the 15 GB).
A repeatability package is available in~\cite{Repeatabilitypackage}.

We conduct two series of experiments:
in Section~\ref{sec:unf}
we study the unfolding performance of the tools; i.e. how fast and how many models can be unfolded
by each tool and how large are the unfolded models, and
in Section~\ref{sec:query} we study the impact of this unfolding on the ability to answer the
queries from the Model Checking Contest 2021.

\subsection{Unfolding experiments} \label{sec:unf}

Table~\ref{tab:unfoldedNets} shows for each of the unfolders
the number of unfolded nets within the memory/time limit.
The last column shows the total number of unfolded nets by all tools combined, and we notice that the combination of methods A$+$B allows us to unfold a superset of models unfolded by the other tools.
Our method A$+$B can unfold 3 nets that
no other tool can unfold; DrinkVendingMachine48, 76, 98. This is directly attributed to method A.

\begin{table}[t]
\centering

\caption{Number of unfolded nets for each unfolder}
\label{tab:unfoldedNets}
\newcolumntype{x}[1]{>{\centering\arraybackslash\hspace{0pt}}p{#1}}

\begin{tabular}{l|x{0.9cm}x{1.1cm}x{0.9cm}x{0.9cm}x{0.9cm}x{0.9cm}x{0.9cm}|x{0.9cm}}
              & Spike & Tapaal & ITS-Tools & A & B  & MCC & A$+$B & Total\\ \hline
Unfolded nets & 170 & 192 & 193 & 195 & 200 & 200  & 203 & 203
\end{tabular}
\end{table}

The comparison of the sizes (total number of transitions and places)
of unfolded nets is done by plotting the ratios
between the size produced by our A+B method and the competing unfolder.
Figure~\ref{Fig:Sizeratio90best} shows the size ratios where at least one comparison is not equal to 1.
We see that our method has a size ratio always smaller or equal to 1 (no other method unfolds
any of the nets to a smaller size compared to our method A+B)
and the size ratio is strictly smaller than 1 for 144 colored nets (out of 213 nets in the database).
Moreover, we can reduce 45 nets by at least one order of magnitute, compared to all other unfolders.

\begin{figure}[t]
\begin{subfigure}[b]{0.5\textwidth}
\centering
\includegraphics[scale=0.53]{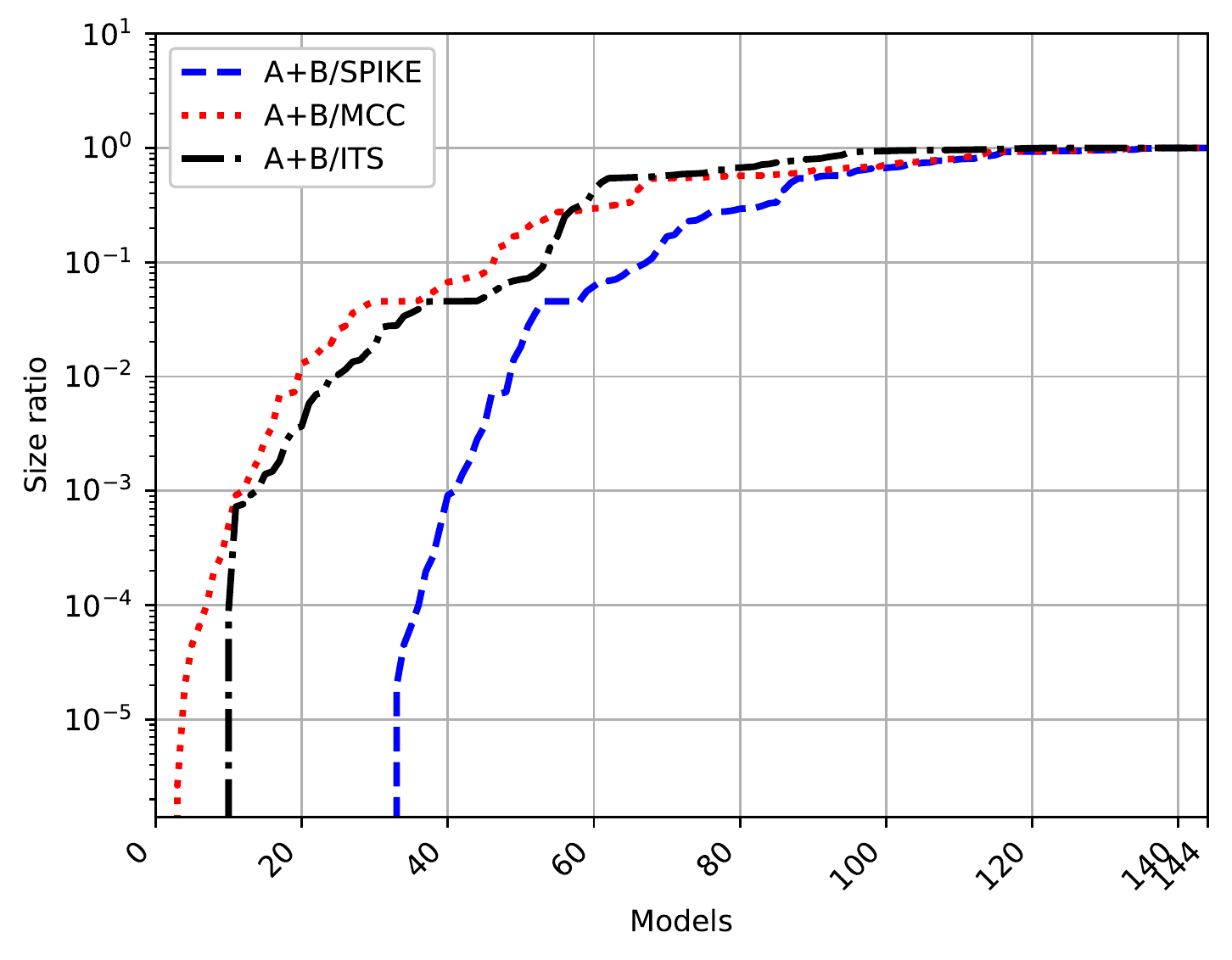}
\caption{Size-ratios (nondecreasingly ordered)}
\label{Fig:Sizeratio90best}
\end{subfigure}
\begin{subfigure}[b]{0.5\textwidth}
\centering
\includegraphics[scale=0.53]{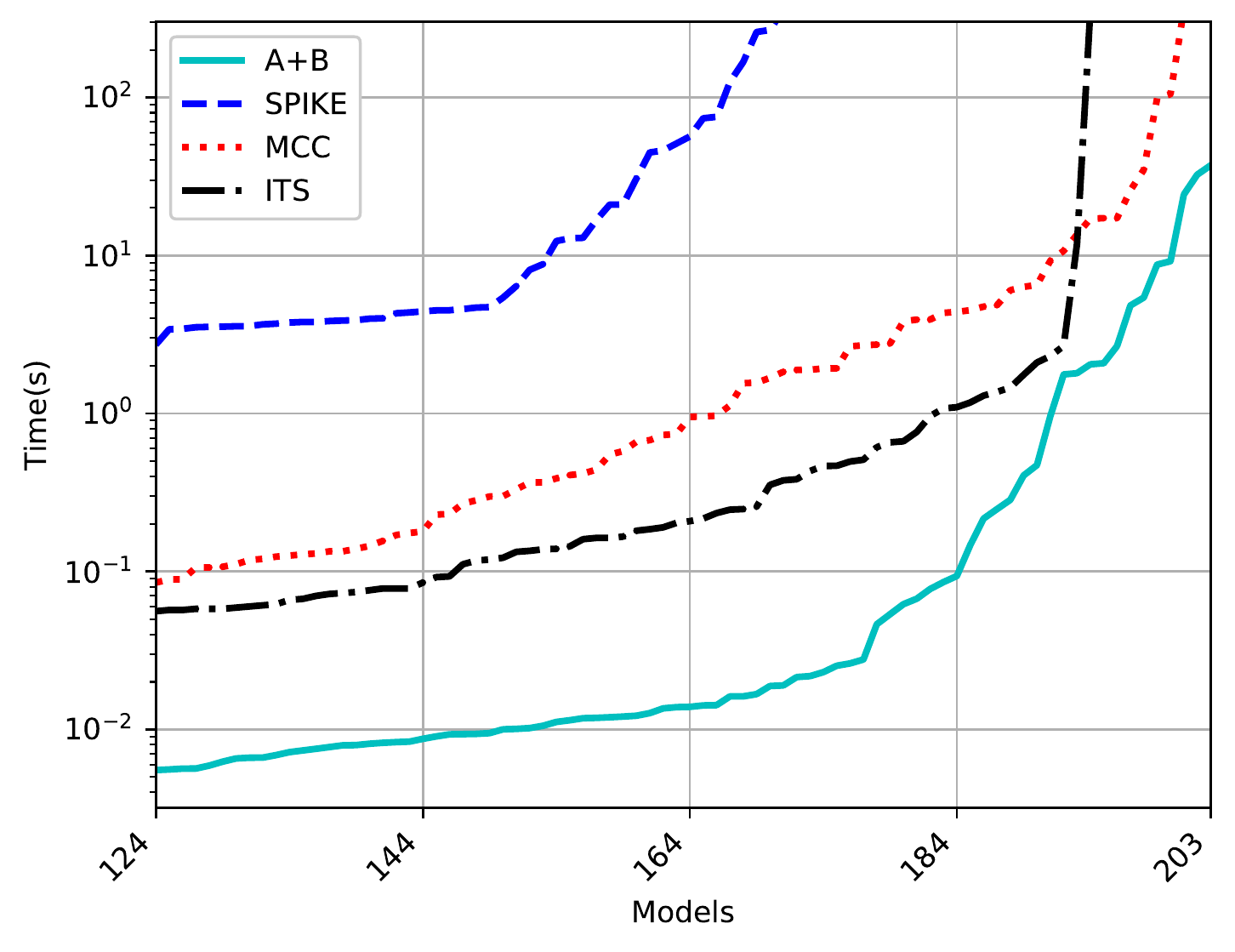}
\caption{Worst 80 unfolding times}
\label{Fig:time80best}
\end{subfigure}\vspace*{-1mm}
\caption{Unfolding size and unfolding time comparison}\vspace*{-2mm}
\end{figure}

As our method outperforms the state-of-the-art unfolders w.r.t. the size of the unfolded nets,
the question is whether the overhead of the advanced static analysis does not kill
the benefits. Fortunately, this is not the case as shown in Figure~\ref{Fig:time80best}
where the 80 slowest unfolding times (independently sorted in nondecreasing order for each tool)
are depicted. The plots show that our method has almost a magnitude reduction in running-time compared to
ITS-Tools and MCC (which are close in performance),
while Spike is significantly slower.
ITS-Tools is generally fast on the nets that are unfolded in less
than 10 seconds, however it becomes gradually slower and has problems unfolding the larger nets.
Our unfolder and the MCC unfolder demonstrate a similar degradation trend in performance.

The overall conclusion is that our advanced analyses requires less overhead compared to other existing
unfolders, while at the same time it significantly decreases the size of the unfolded nets.

\subsection{Query verification experiments} \label{sec:query}

In these experiments, we examine which unfolding engine allows for most query answers verified
on the unfolded nets. To allow for a fair comparison, we let each tool unfold and output the net to a PNML file. Regarding queries, both method A$+$B and ITS-Tools can already output the unfolded queries. For MCC, we implement our own translation from the colored queries to the unfolded queries for the given nets. For Spike,
we were not able to construct a query unfolder that worked consistently and for this
reason Spike is excluded from the query verification experiments.

\begin{table}[!b]
\vspace*{-2mm}
\centering
\caption{Number of queries answered for the unfolded nets of each tool. The \% column describes how many percent of the total available queries in each category were answered. Each category counts $213 \cdot 16=3408$ queries.}
\label{tab:queriesAnswered} \vspace*{-1mm}
\scalebox{0.9}{
\begin{tabular}{lllllll}
\hline
\multicolumn{7}{c}{\textbf{Cardinality Queries}}                                                                                                                                                     \\ \hline
\multicolumn{1}{l|}{}                               & \multicolumn{2}{c|}{A+B}                              & \multicolumn{2}{c|}{MCC}                              & \multicolumn{2}{c}{ITS-Tools}         \\ \hline
\multicolumn{1}{l|}{}                               & \multicolumn{1}{l|}{Solved} & \multicolumn{1}{c|}{\%} & \multicolumn{1}{l|}{Solved} & \multicolumn{1}{c|}{\%} & \multicolumn{1}{l|}{Solved} & \multicolumn{1}{c}{\%} \\ \hline
\multicolumn{1}{l|}{ReachabilityCardinality}        & \multicolumn{1}{l|}{2909}   & \multicolumn{1}{l|}{85.4}   & \multicolumn{1}{l|}{2793}   & \multicolumn{1}{l|}{82.0}   & \multicolumn{1}{l|}{2759}   &  81.0 \\ \hline
\multicolumn{1}{l|}{CTLCardinality}                 & \multicolumn{1}{l|}{2651}   & \multicolumn{1}{l|}{77.8}   & \multicolumn{1}{l|}{2490}   & \multicolumn{1}{l|}{73.1}   & \multicolumn{1}{l|}{2443}   &  73.8  \\ \hline
\multicolumn{1}{l|}{LTLCardinality}                 & \multicolumn{1}{l|}{2952}   & \multicolumn{1}{l|}{86.6}   & \multicolumn{1}{l|}{2818}   & \multicolumn{1}{l|}{82.7}   & \multicolumn{1}{l|}{2785}   & 81.7   \\ \hline
\multicolumn{1}{l|}{\textbf{Total}}                 & \multicolumn{1}{l|}{\textbf{8512}}       & \multicolumn{1}{l|}{\textbf{83.3}}   & \multicolumn{1}{l|}{\textbf{8101}}       & \multicolumn{1}{l|}{\textbf{79.2}}   & \multicolumn{1}{l|}{\textbf{7987}}       &    \textbf{78.1}\\ \hline
\multicolumn{7}{c}{\textbf{Fireability Queries}}                                                                                                                                                                 \\ \hline
\multicolumn{1}{l|}{ReachabilityFireability}        & \multicolumn{1}{l|}{2567}   & \multicolumn{1}{l|}{75.3}   & \multicolumn{1}{l|}{2484}   & \multicolumn{1}{l|}{72.9}   & \multicolumn{1}{l|}{2513}   &  73.7  \\ \hline
\multicolumn{1}{l|}{CTLFireability}                 & \multicolumn{1}{l|}{2047}   & \multicolumn{1}{l|}{60.1}   & \multicolumn{1}{l|}{1878}   & \multicolumn{1}{l|}{55.1}   & \multicolumn{1}{l|}{1695}   &  49.7  \\ \hline
\multicolumn{1}{l|}{LTLFireability}                 & \multicolumn{1}{l|}{2798}   & \multicolumn{1}{l|}{82.1}   & \multicolumn{1}{l|}{2639}   & \multicolumn{1}{l|}{77.4}   & \multicolumn{1}{l|}{2520}   & 73.9  \\ \hline
\multicolumn{1}{l|}{\textbf{Total}}                 & \multicolumn{1}{l|}{\textbf{7412}}       & \multicolumn{1}{l|}{\textbf{72.5}}   & \multicolumn{1}{l|}{\textbf{7001}}       & \multicolumn{1}{l|}{\textbf{68.5}}   & \multicolumn{1}{l|}{\textbf{6728}}       &  \textbf{65.8}  \\ \hline \hline
\multicolumn{1}{l|}{\textbf{Total query answers}} & \multicolumn{1}{l|}{\textbf{15924}}       & \multicolumn{1}{l|}{\textbf{77.9}}   & \multicolumn{1}{l|}{\textbf{15102}}       & \multicolumn{1}{l|}{\textbf{73.9}}   & \multicolumn{1}{l|}{\textbf{14715}}       &  \textbf{72.0}  \\ \hline
\end{tabular} }
\end{table}

Since we are testing the effect of the unfolding and not the verification engine, we use \textit{verifypn}
(revision 507c8ee0) to verify the queries on the nets unfolded by the different unfolders.
There is a total of 20,448 queries to be answered. The results can be seen in Table~\ref{tab:queriesAnswered}.

We see that using the method A$+$B to unfold nets allows us to answer more queries in every category due to the generally smaller nets it unfolds to. In total, we are able to answer $4$ percentage points more queries using the unfolded nets of method A$+$B compared to using the unfolded nets of MCC and $5.9$ percentage points more compared to ITS-Tools.

\section{Conclusion}
We presented two complementary methods 
for reducing the unfolding size of colored Petri nets (CPN). Both methods are proved
correct and implemented in an open-source verification engine of the tool TAPAAL.
Experimental results show a significant improvement in the size of unfolded
nets, compared to state-of-the-art tools, without compromising the unfolding speed.
The actual verification on the models and queries from the 2021 Model Checking Contest
shows that our unfolding technique allows us to solve 4\% more queries compared to the
second best competing tool. In future work, we plan to combine our approach with
structural reduction techniques applied directly to the colored nets. 

\paragraph{Acknowledgments.} We would like to thank Yann Thierry-Mieg
for his answers and modifications to the ITS-Tools,
Silvano Dal Zilio for his answers/additions concerning the MCC unfolder
and Monika Heiner and Christian Rohr for their answers concerning the tools Snoopie, Marcie and Spike.

\sloppy
\bibliographystyle{fundam}
\bibliography{references}
\end{document}